\documentclass[final,10pt,journal,twoside]{IEEEtran}

\usepackage{microtype}
\usepackage{siunitx}
\usepackage{graphics, graphicx}
\usepackage{amsthm, amssymb}
\usepackage{ifpdf}
\usepackage{epstopdf}
\usepackage{cite}
\usepackage{algorithm, algpseudocode, algpascal}
\usepackage{array}
\usepackage{booktabs}
\usepackage{stfloats}
\usepackage{enumerate}
\usepackage{amsfonts}
\usepackage{eurosym}
\usepackage{float}
\usepackage{makeidx}
\usepackage{graphicx}
\usepackage{amsmath}
\newtheorem{lemma}{Lemma}
\newtheorem{theorem}{Theorem}
\newtheorem{definition}{Definition}
\newtheorem{remark}{Remark}
\newtheorem{corollary}{Corrollary}
\newtheorem{assumption}{Assumption}
\input{tcilatex}
\allowdisplaybreaks[4]

\begin{document}

%
\title{The system dynamics analysis, resilient and fault-tolerant control for cyber-physical systems}
\author{Linlin~Li,~\IEEEmembership{Senior Member,~IEEE}, Steven~X.~Ding, Liutao Zhou, Maiying Zhong, and Kaixiang Peng 
\thanks{%
This work has been supported by the National Natural Science Foundation of
China under Grants 62322303, 62233012,  and U21A20483.}
\thanks{L. Li and K. Peng  are with School of Automation and Electrical Engineering, University of Science and Technology Beijing, Beijing 100083, P. R. China, Email: linlin.li@ustb.edu.cn, kaixiang@ustb.edu.cn.}
 \thanks{S. X. Ding and L. Zhou are with the Institute for Automatic Control and Complex Systems (AKS), University of Duisburg-Essen,  Germany.  Email: steven.ding@uni-due.de, liutao.zhou@uni-due.de.}
 \thanks{M. Zhong is with the College of Electrical Engineering and Automation,
Shandong University of Science and Technology, Qingdao 266590, China.  Email: mzhong@buaa.edu.cn.}}


\maketitle

\begin{abstract}
This paper is concerned with the detection, resilient and fault-tolerant control issues for cyber-physical systems. To this end, the impairment of system dynamics caused by the defined types of cyber-attacks and process faults is analyzed. Then, the relation of the system input and output signals with the residual  subspaces spanned by both the process and the controller is studied. Considering the limit capacity of standard observer-based detection and feedback control schemes in detecting and handling the cyber-attacks, a modified  configuration for cyber-physical systems is developed by transmitting the combinations of  the input and output residuals instead of the input and output signals, which is facile for dealing with both the process faults and cyber-attacks. It is followed by the integrated design of fault and attack detection, resilient and fault-tolerant control schemes.
To enhance the detectability of cyber-attacks, the potential stealthy attack mechanisms on deteriorating the tracking behavior and feedback control performance are developed from the attackers' point of view, and the associated detection schemes for such stealthy attacks are proposed from the defenders' point of view. 
A case study on the robotino system is utilized to demonstrate the proposed resilient cyber-physical configuration. 
\end{abstract}


\section{Introduction}
\IEEEPARstart{T}{oday's}
automatic control systems as the centrepiece of
industrial cyber-physical systems (CPSs) are fully equipped with intelligent
sensors, actuators and an excellent information infrastructure. 
It is a
logic consequence of ever increasing demands for system performance and
production efficiency that today's automatic control systems are of an
extremely high degree of integration, automation and complexity 
\cite{Ding2020}. Maintaining
reliable and safe operations of automatic control systems are of elemental
importance for optimally managing industrial CPSs over the whole operation
life-cycle. As an indispensable maintenance functionality, real-time
monitoring, fault detection (FD) and fault-tolerant control (FTC) are widely integrated in
automatic control systems and runs parallel to the embedded control systems \cite{Ding2014,CRB2001,PFC00}.
In a traditional automatic control system, FD and FTC were mainly dedicated to maintaining functionalities of
sensors and actuators as the key components embedded in the system with technical faults \cite{AET_SMO_book_2011,Ding2020}.  

 Recently, a new
type of malfunctions, the so-called cyber-attacks on automatic control
systems, have drawn attention on the urgent need for developing new
monitoring, diagnosis and resilient control strategies \cite{DeruiDing2021-survey,TGXHV2020,TEIXEIRA-zero-attack_2015,Zhou2021IEEE-Proc}. Cyber-attacks can not
only considerably affect functionalities of sensors and actuators, but also
impair communications among the system components and sub-systems. Different from
technical faults, cyber-attacks are artificially created and could be
designed by attackers in such a way that they cannot be detected using the
existing diagnosis techniques  and lead to immense damages during system operations   \cite{Survey-attack-detection2018,Zhang2020attacks}. Such cyber-attacks are called stealthy. A
further type of cyber-attacks are the so-called eavesdropping attacks \cite{DIBAJI2019-survey}.
Although such attacks do not cause changes in system dynamics and
performance degradation, they enable adversary to gain system knowledge
which can be used to design, for instance, stealthy attacks  \cite{Chen2018attacks,Zhang2020attacks,Yang2022attacks,Shang2022attacks}.
As a result, an early and reliable detection of 
cyber-attacks is becoming a vital requirement on cyber-security of industrial
CPSs \cite{MMM2020}. In the
literature, a great number of results have been reported about detecting the
so-called replay, zero dynamics, covert attacks, false data injection
attacks or kernel attacks, which are stealthy integrity attacks and cannot,
structurally, be detected by means of a standard observer-based fault
detection system \cite%
{DIBAJI2019-survey,TGXHV2020,DLautomatica2022}. 
Specifically, alternative detection schemes have been proposed for the detection of stealthy integrity attacks in the so-called unified control and detection framework \cite{DLautomatica2022}. By extending the residual detection to both the system input and output subspaces, the stealthy  attack can be well detected.
On the other hand, innumerable capable diagnosis of both the physical faults and attacks have been developed with various specifications  in recent years \cite{HDCL-TCSII-2024,ZKPP-CSL-2021,RA-EJC-2023}. 
To our best knowledge, limit attention has been made on the analysis of the influence of the physical faults and attacks on the systems dynamics and distinguish them structurally.

As a
response to the security issues of CPSs,
resilient control has received increasing attention from both the research and application domains. As appropriate control-theoretic countermeasures, the
so-called resource-aware secure control methods \cite{DeruiDing2021-survey},
like event-triggered control algorithms or switching control mechanisms, are
prevalent resilient control schemes. Roughly speaking, these methods make
the CPSs highly resilient against attacks by means of optimal management of
data communications among subsystems in the CPS. On the other hand, the well-established control techniques have been widely implemented in resilient control, for instance model predictive control \cite{SZCDCX-AUTO-2024}, adaptive control \cite{AY-IS-2018}, and sliding mode control \cite{YZDYF-AUTO-2023}. Most of the existing methods require the precise information of the states of CPSs. To this end, the resilient state estimation has been studied by applying the Kalman filter, Luenberger observer  and extended state observer as the tools \cite{DeruiDing2021-survey}, which dominates the thematic area of resilient control of CPSs.
An immediate consequence is that the proposed design algorithms for resilient control are limited to updating the existing methods and
schemes for  FTC \cite{Ferreira-survey-2024}. 
In recent years, the integrated design of FD and FTC is state of the art of the advanced FTC technique. Reviewing the recent publications shows that  most of the studies on resilient control of CPSs address detection, estimation and control issues separately.
Our work is motivated by the above observation and in particular driven by the question: What is difference of the impact on the system dynamics  for cyber-attacks and technical faults? What is the difference between FD and attack detection, and further FTC and resilient control? Is it possible to develop an alternative CPS configuration to achieve the detection, resilient and FTC for both the cyber-attacks and technical faults in the integrated fashion? 

%


In the recent decade, a trend of investigating attack mechanisms can be observed, which is not only helpful to assess the security weakness and vulnerability of CPSs, but also useful in designing the defense strategy \cite{SMMSF-TAC-2021,ZHOU2023110723}. Among the involved studies, the denial-of-service (DoS) attacks \cite{QLSY-TAC-2018} and deception attacks \cite{AY-TAC-2018}, as the main categories of malicious attacks,  attracted the major attention. By blocking the communication links of system components over wireless networks, DoS attacks can degrade the system performance \cite{GSJS-AUTO-2018,RWDS-AUTO-22018}. The deception attacks are usually designed by the attacks to cause system performance degradation
through modifying data packets without being detected. Among the involved studies, switched strategy \cite{WSC-TC-2018}, multi-sensor joint attack \cite{FQCTZ-TAC-2020,GYH-TCNS-2021}, data-driven method \cite{AY-TAC-2018,ZHZL-TC-2021}, built the main stream. 
 Nevertheless, the major attention has been paid to degrade the tracking performance, enlarge the state estimation error, and increase  linear quadratic Gaussian (LQG) control performance \cite{JY-TC-2024,SCZC-TCNS-2022,ZLTLCX-TAC-2023,LY-TAC-2022}, while limit research efforts have been made on deteriorating the feedback control performance \cite{MS-TAC-2016}. 
On the other hand, most of the existing attack design schemes are implemented based on some strict assumptions, which limits the applications.


Inspired by these observations, the main objective of this paper is to deal with the fault and  cyber-security  issues in a systematic and reliable manner.
The main contributions can be summarized as follows:
\begin{itemize}
\item  The systematic analysis of the closed-loop dynamics of CPSs under technical faults and defined types of cyber-attacks, and the attack- and fault-induced impairment of system dynamics are given for the first time. 
\item The parameterization form of the closed-loop dynamics is studied first, and the relation of system signal subspace with the residual subspaces spanned by both the input and output signal is established. 
\item A modified CPS configuration is developed by utilizing the combinations of  the input and output residuals as the transmitted data instead of the process input and output to ensure i) the cyber-security in case that the communications fail, ii)  detection of cyber-attacks and technical faults, and iii) the data privacy.
\item The integrated design scheme for fault and attack detection, FTC and  resilient control is investigated for the first time  based on the proposed CPS configuration. 

\item The stealthy attack design schemes are proposed for deteriorating  the tracking behavior and feedback control performance respectively from  the attackers' point of view. Then, from the defenders' point of view, two performance degradation monitoring-based detection methods of stealthy attacks are developed.

\end{itemize}

\noindent\textbf{Notations}.  $\mathcal{H}_2^m$ represents the signal space of all the signals  of dimension $m$ with bounded energy. $\mathcal{RH}_\infty$ denotes the space of all rational transfer functions of stable systems. The root mean square (RMS) of signal $\alpha$ over the time interval $[k+1,k+\tau]$ is defined as $
\left\Vert \alpha\right\Vert _{RMS,\left[ k+1,k+\tau \right] }:=\sqrt{%
\frac{1}{\tau }\sum\nolimits_{i=k+1}^{k+\tau }\alpha^{T}\left( i\right)
\alpha\left( i\right) }.
$

\section{Preliminaries and Problem Formulation}\label{sec2}
\subsection{System factorizations}
Consider the  following linear discrete-time system
\begin{align}
&x(k+1)=A x(k)+Bu(k)+E_ff+w(k)\notag\\
&y(k)=C x(k)+Du(k)+F_ff+\nu(k)\label{eq2-b}
\end{align}
where $y\hspace{-2pt}\in \hspace{-2pt}\mathcal{R}^{k_y}\hspace{-1pt},u\hspace{-2pt}\in \hspace{-2pt}\mathcal{R}^{k_u}\hspace{-1pt},x\hspace{-2pt}\in \hspace{-2pt}\mathcal{R}^{k_x}$ denote the output, input and  state, respectively. $w(k),\nu(k)$ represent the process and measurement noise with $w\hspace{-2pt}\sim\hspace{-2pt} (0,\Sigma_w), \nu\hspace{-2pt}\sim\hspace{-2pt} (0,\Sigma_{\nu})$.    $A,B,C,D$ are system matrices with appropriate dimensions.  $f$ denotes the technical fault with $E_f, F_f$ as the distribution matrices.

Let $G=(A,B,C,D)$. $G(z)=N(z)M^{-1}(z)=\hat{M}^{-1}(z)\hat{N}(z)$ are called the right coprime factorization (RCF) and left coprime factorization (LCF) of $G(z)$ respectively if there exists $X(z), Y(z), \hat{X}(z),\hat{Y}(z)\in \mathcal{RH}_\infty$ such that the following Bezout identity holds
\begin{align}
\left[ 
\begin{array}{cc}
X(z) & Y(z)\\
-\hat{N}(z) & \hat{M}(z) 
\end{array}%
\right]\left[ 
\begin{array}{cc}
M(z) & -\hat{Y}(z)\\
N(z) & \hat{X}(z) 
\end{array}%
\right]=I\label{eq-3a}
\end{align}
where $M(z), N(z), \hat{M}(z),\hat{N}(z)\in \mathcal{RH}_\infty$.
The state space representations for the associated transfer matrices in Bezout identity are given by 
\begin{align}
&\hat{M}\hspace{-3pt}=\hspace{-3pt}\left(A\hspace{-3pt}-\hspace{-3pt}LC,-\hspace{-2pt}L,WC,W\right)\hspace{-2pt}, \hat{N}\hspace{-3pt}=\hspace{-3pt}\left(A\hspace{-3pt}-\hspace{-3pt}LC,B\hspace{-3pt}-\hspace{-3pt}LD,WC,WD\right)\notag\\
&{M}\hspace{-2pt}=\hspace{-2pt}\left(A\hspace{-2pt}+\hspace{-2pt}BF,BV,F,V\right)\hspace{-2pt},{N}\hspace{-2pt}=\hspace{-2pt}\left(A\hspace{-2pt}+\hspace{-2pt}BF,BV,C\hspace{-2pt}+\hspace{-2pt}DF,DV\right)\notag\\
&\hat{X}\hspace{-3pt}=\hspace{-3pt}\left(A\hspace{-3pt}+\hspace{-3pt}BF,L,C\hspace{-3pt}+\hspace{-3pt}DF,W^{-1}\right),\hat{Y}\hspace{-3pt}=\hspace{-3pt}\left(A\hspace{-3pt}+\hspace{-3pt}BF,-\hspace{-2pt}LW^{-1},F,0\right)\notag\\
&X\hspace{-3pt}=\hspace{-3pt}\left(A\hspace{-3pt}-\hspace{-3pt}LC,-\hspace{-2pt}(B\hspace{-3pt}-\hspace{-3pt}LD),F,I\right)\hspace{-2pt}, Y\hspace{-3pt}=\hspace{-3pt}\left(A\hspace{-3pt}-\hspace{-3pt}LC,\hspace{-3pt}-\hspace{-2pt}L,F,0\right)\label{eq-TF-MN}
\end{align}
where $F,L$ are selected such that $A-LC$ and $A+BF$ are Schur matrices. $W$ and $V$ are invertible matrices \cite{Ding2014,Zhou96}. 

According to Youla parameterization, all the controllers that can stabilize the plant $G(z)$ can be parameterized by 
\begin{equation}\label{eq-Youla}
\setlength{\abovedisplayskip}{5pt}\setlength{\belowdisplayskip}{5pt}
K\hspace{-2pt}=\hspace{-3pt}-\hspace{-2pt}(\hat{Y}\hspace{-2pt}-\hspace{-2pt}MQ)(\hat{X}\hspace{-2pt}+\hspace{-2pt}NQ)^{-1}\hspace{-2pt}=\hspace{-3pt}-\hspace{-2pt}(X\hspace{-2pt}+\hspace{-2pt}Q\hat{N})^{-1}(Y\hspace{-2pt}-\hspace{-2pt}Q\hat{M})
\end{equation}
where $Q\in \mathcal{RH}_\infty$ is the parameter system \cite{Zhou96}.


\subsection{The signal subspaces}

The input and output (I/O) signal space of the system (\ref{eq2-b})  can be described by the following definition. 
\begin{definition}
Consider the system (\ref{eq2-b}) with the RCF and LCF as $G(z)=N(z)M^{-1}(z)=\hat{M}^{-1}(z)\hat{N}(z)$. 
The kernel and image subspaces of the system (\ref{eq2-b}) are defined by
\begin{align}
&\mathcal{K}_{G}=\left\{\left[ 
\begin{array}{c}
u\\y
\end{array}%
\right]\in \mathcal{H}_2^{k_u+k_y}:\left[ 
\begin{array}{cc}
-\hat{N}& \hat{M}
\end{array}%
\right]\left[ 
\begin{array}{c}
u\\y
\end{array}%
\right]=0\right\}\\
&\mathcal{I}_{G}=\left\{\left[ 
\begin{array}{c}
u\\y
\end{array}%
\right]\in \mathcal{H}_2^{k_u+k_y}:\left[ 
\begin{array}{c}
u\\y
\end{array}%
\right]\hspace{-2pt}=\hspace{-2pt}\left[ 
\begin{array}{c}
M \\ N
\end{array}%
\right]\hspace{-2pt}v, v\in \mathcal{H}_2^{k_u}\right\}.\label{eq-SIR}
\end{align}
\end{definition}

It is well known that the LCF can be applied as the residual generator which is essential in fault detection
\begin{equation}
\setlength{\abovedisplayskip}{3pt}\setlength{\belowdisplayskip}{3pt}
r_y(z)=-\hat{N}(z)u(z)+\hat{M}(z)y(z)\end{equation}
where $r_y$ represents the residual signal delivering the information for the uncertainties and potential faults. 
The state-space representation can be written into the observer-based form
\begin{align}
&\hat{x}(k+1)=A\hat{x}(k)+Bu(k)+LW^{-1}r_y(k)\notag\\
&r_{y}(k)=W(y(k)-\hat{y}(k)),\hat{y}(k)=C\hat{x}(k)+Du(k)\label{eq-observer}
\end{align}
where $\hat{x}, \hat{y}$ denotes the state and output estimation, respectively.

%
%
%
Recall that due to Bezout identity,
\begin{equation*}
\setlength{\abovedisplayskip}{3pt}\setlength{\belowdisplayskip}{3pt}
\forall \left[ 
\begin{array}{c}
u\\y
\end{array}%
\right]\in \mathcal{I}_G, r_y=\left[ 
\begin{array}{cc}
-\hat{N} & \hat{M}
\end{array}%
\right]\left[ 
\begin{array}{c}
u\\y
\end{array}%
\right]=0.
\end{equation*}
It is of interest to notice that  the residual $r_y\neq 0$ delivers the  information for the potential uncertainties and faults in the process.
On the basis of the Bezout identity, it turns out
\begin{equation*}
\setlength{\abovedisplayskip}{3pt}\setlength{\belowdisplayskip}{3pt}
\forall r_y\neq 0, r_y(z)
\hspace{-2pt}=\hspace{-2pt}\left[ \hspace{-2pt}
\begin{array}{cc}
-\hspace{-2pt}\hat{N}(z) & \hat{M}(z)
\end{array}%
\right]\hspace{-2pt}\left[ \hspace{-2pt}
\begin{array}{c}
-\hspace{-2pt}\hat{Y}(z)+M(z)Q(z)\\ \hat{X}(z)+N(z)Q(z)
\end{array}%
\hspace{-2pt}\right]\hspace{-2pt}r_y(z).
\end{equation*}
Recall that $(\hat{X}(z)+N(z)Q(z),-\hat{Y}(z)+M(z)Q(z))$ build the LCF of the controller $K(z)$.
Consequently, any process I/O data which ensures  $r_y\neq 0$ can be characterized by the image subspace of the controller. 
\begin{definition}
The image subspace of controller (\ref{eq-Youla}) is given by
\begin{equation}
\setlength{\abovedisplayskip}{3pt}\setlength{\belowdisplayskip}{3pt}
\mathcal{I}_K\hspace{-2pt}=\hspace{-2pt} \left\{\left[ \hspace{-2pt}
\begin{array}{c}
u\\y
\end{array}%
\hspace{-2pt}\right]\hspace{-2pt}\in\hspace{-2pt} \mathcal{H}_2^{k_u+k_y}, \left[ \hspace{-2pt}
\begin{array}{c}
u\\y
\end{array}%
\hspace{-2pt}\right]\hspace{-2pt}=\hspace{-2pt}\left[ 
\begin{array}{c}
-\hspace{-2pt}\hat{Y} \hspace{-2pt}+\hspace{-2pt}MQ\\ \hat{X}\hspace{-2pt}+\hspace{-2pt}NQ\hspace{-2pt}
\end{array}%
\right]r_y, r_y\hspace{-2pt}\in \hspace{-2pt}\mathcal{H}_2^{k_y}\right\}.\notag
\end{equation}
\end{definition}

\subsection{Problem formulation}



With rapidly
increasing threats of cyber-crime and the resulted damages to industrial
CPSs, cyber-security imposes new challenging issues in research of
 diagnosis and resilient control of both cyber-attacks and technical faults,
which should be managed with
high reliability and in an integrated fashion with the existing control and
monitoring algorithms and facilities.
Thus, the main objective of this paper is to develop a detection, resilient and fault-tolerant control framework in a systematic and reliable manner. 
To this end, our efforts in the first part of this paper have been mainly made
 on
\begin{itemize}
\item analyzing 
the attack- and fault-induced impairment of system dynamics of a general type of CPSs, and
\item discussing the potential of the existing control and detection
methods to manage the system performance impairment due to the cyber-attacks and technical faults.
\end{itemize}
It can be shown that the standard observer-based detection and feedback control schemes are limited in dealing with cyber-attacks.  Thus, it is the main objective of the second part of this paper to develop a modified CPS configuration which guarantee not only high resilience of CPSs under cyber-attacks and high fault tolerance against process faults, but also ensure ``fail-safe" cyber-security and data privacy. 
Then the integrated design schemes  for i) fault detection, ii) attack detection, iii) resilient and fault-tolerant control are studied. To further enhance the security of the CPSs, the stealthy attack design schemes are
proposed for deteriorating  the tracking behavior and feedback control performance respectively from  the attackers' point of view. Then, from the defenders' point of view, the associated performance degradation monitoring-based detection methods of stealthy attacks are developed.

\section{Analysis of System Dynamics under Cyber-Attacks and  Faults}\label{sec3}
In this section, the system dynamics under cyber-attacks and faults for a general type of CPS systems are analyzed.

Consider a CPS system modelled by (\ref{eq2-b}). 
  The control signal $u_{MC}$ is generated at the control station as
\begin{align}
&u_{MC}(z)=K\left( z\right) y^{a}(z)+v(z)\label{eq15-1a}
\end{align}%
and sent to the plant as $u(z)=u_{MC}^{a}(z)$ 
with the following general type of cyber-attacks under consideration
\begin{align}
y^{a}(k)& \hspace{-2pt}=\hspace{-2pt}y(k)\hspace{-2pt}+\hspace{-2pt}a_{y}(k),  u_{MC}^{a}(k) \hspace{-2pt}=\hspace{-2pt}u_{MC}(k)\hspace{-2pt}+\hspace{-2pt}a_{u_{MC}}(k)\label{eq14-4b}
\end{align}%
where $a_{y}$ and $a_{u_{MC}}$ represent the cyber-attacks
injected into the output and input communication channels, respectively. 
$K(z)$  is the feedback gain matrix given by  (\ref{eq-Youla})  and $v(z)$ denotes the reference signal.

\subsection{Impairment of system dynamics caused by cyber-attacks and faults} 

We now analyze the possible impairment of the system dynamics under the faults and the  cyber-attacks (\ref{eq14-4b}).  
To ease the presentation, an alternative representation of the plant (\ref{eq2-b}) is studied first. 
\begin{theorem}
The process (\ref{eq2-b}) can be equivalently described by 
\begin{equation}
y(z)=G(z)u(z)+\hat{M}^{-1}(z)r_{y}(z)  \label{eq14-0a} 
\end{equation}
with $r_y(z)$ generated by (\ref{eq-observer}).
\end{theorem}
\begin{proof}
It is evident that the process (\ref{eq2-b}) can be described by the following kernel-based system I/O model \cite{Ding2020}
\begin{equation}
y(k)=\hat{y}(k)+W^{-1}r_{y}(k).%
 \label{eq14-0}
\end{equation}%
 Observe that this model can be
equivalently written into 
\begin{align*}
y(z)& =G(z)u(z)+\left( I+\left( zI-A\right) ^{-1}L\right) W^{-1}r_{y}(z).
\end{align*}%
It is of interest to notice that 
\begin{equation}
\setlength{\abovedisplayskip}{3pt}\setlength{\belowdisplayskip}{3pt}
\left( I+\left( zI-A\right) ^{-1}L\right) W^{-1}=\hat{M}^{-1}(z)  \notag 
\end{equation}%
which completes the proof.
\end{proof}
\begin{remark}
It is noteworthy that the kernel-based system I/O-model (\ref{eq14-0a})  describes the system dynamics without approximation, independent of  the existence of any type of uncertainties and plant faults. Thus, throughout of the paper, (\ref{eq14-0a}) is adopted to represent the process  (\ref{eq2-b}). 
\end{remark}


\begin{theorem}\label{th3}
Given CPS modelled by (\ref{eq15-1a})-(\ref{eq14-0a}),  the system dynamics under  the attacks and faults  are governed by 
\begin{align}
&\left[ 
\begin{array}{c}
u \\ 
y%
\end{array}%
\right] =\left[ 
\begin{array}{c}
M \\ 
N%
\end{array}%
\right] \left( \bar{v}+d_{a}\right) +\left[ 
\begin{array}{c}
-\hat{Y}+MQ \\ 
\hat{X}+NQ%
\end{array}%
\right] r_{y},  \label{eq14-10}
\\
&\bar{v}\hspace{-2pt}=\hspace{-2pt}\left( X\hspace{-2pt}+\hspace{-2pt}Q\hat{N}\right) v,d_a\hspace{-2pt}=\hspace{-2pt}\left( \hspace{-1pt}X\hspace{-2pt}+\hspace{-2pt}Q\hat{N}\hspace{-1pt}\right) a_{u_{MC}}\hspace{-2pt}-\hspace{-2pt}\left( \hspace{-2pt}Y\hspace{-2pt}-\hspace{-2pt}Q\hat{M}\hspace{-2pt}\right)\hspace{-2pt}a_y.\notag
\end{align}
\end{theorem}
\begin{proof}
It is straightforward that
\begin{align}
&u=u_{MC}+a_{u_{MC}}=Ky+v+d,d=Ka_y+a_{u_{MC}}. \label{eq-feedback11}
\end{align}
 Considering the process model (\ref{eq14-0a}) and (\ref{eq-feedback11}), together with the coprime factorization of $G,K$, it is easy to see
\begin{align*}
\left[ 
\begin{array}{c}
u \\ 
y%
\end{array}%
\right] 
& =\left[ 
\begin{array}{cc}
I & \text{ }-K \\ 
-G & \text{ }I%
\end{array}%
\right] ^{-1}\left[ 
\begin{array}{c}
v+d \\ 
\hat{M}^{-1}r_{y}%
\end{array}%
\right] \\
& =\left[ 
\begin{array}{cc}
X+Q\hat{N} & \text{ }Y-Q\hat{M} \\ 
-\hat{N} & \text{ }\hat{M}%
\end{array}%
\right] ^{-1}\left[ 
\begin{array}{c}
\left( X+Q\hat{N}\right) \left( v+d\right) \\ 
r_{y}%
\end{array}%
\right] \\
& =\left[ 
\begin{array}{cc}
M & \text{ }-\hat{Y}+MQ \\ 
N & \text{ }\hat{X}+NQ%
\end{array}%
\right] \left[ 
\begin{array}{c}
\left( X+Q\hat{N}\right) \left( v+d\right) \\ 
r_{y}%
\end{array}%
\right]
\end{align*}
which completes the proof.
\end{proof}

According to (\ref{eq-feedback11}), the equivalent configuration of the CPS under cyber-attacks is depicted in Fig. \ref{fig2} with $\bar{r}_y=\hat{M}^{-1}r_y$. It is of interest to notice that the impacts of the cyber-attacks on the
closed-loop dynamics are equivalently modelled  as an unknown input added
to the reference signal.
 
\begin{figure}[!t]
\centering\includegraphics[width=5.5cm]{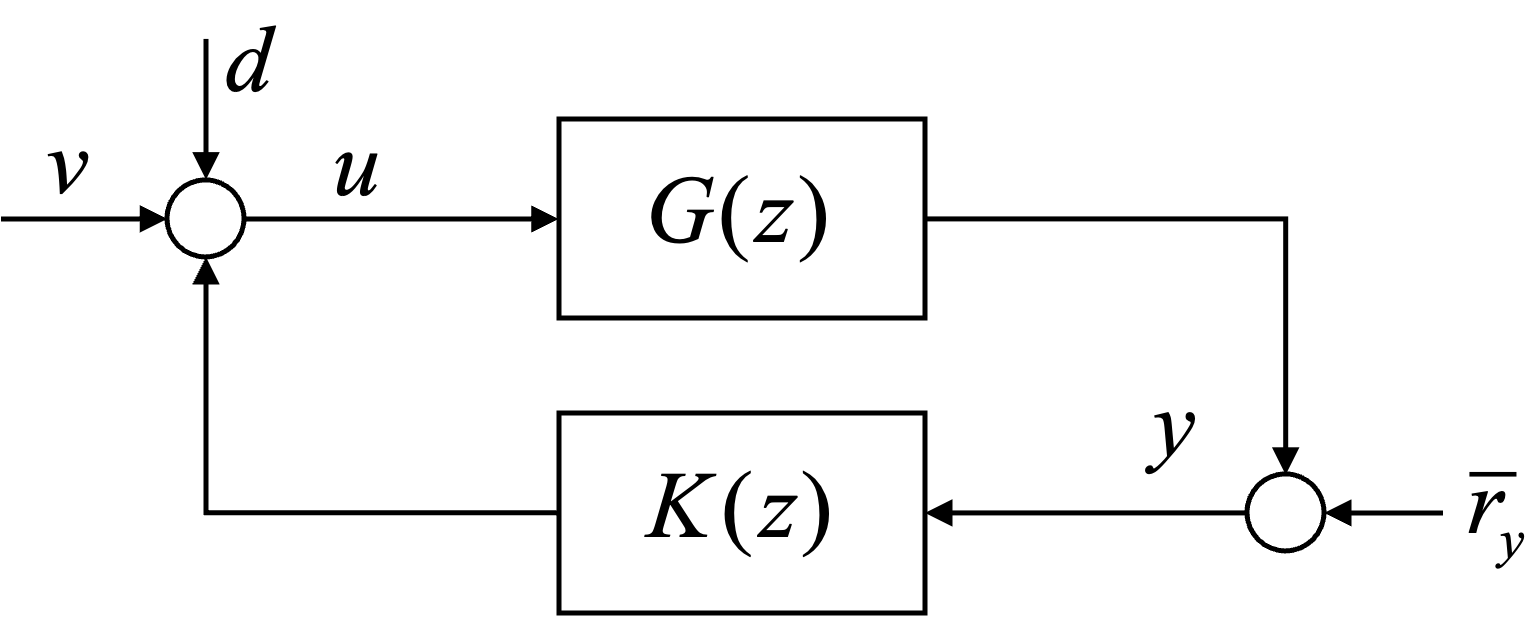}
\caption{The equivalent system dynamics for CPS systems with cyber-attacks}
\label{fig2}
\end{figure}

In what follows, we further analyze the system dynamics under the cyber-attacks  $a_{y}$ and $a_{u_{MC}}$ modelled by 
\begin{align}
\left[ 
\begin{array}{c}
a_{u_{MC}}(z) \\ 
a_y(z)
\end{array}%
\right]=\Pi^a(z)\left[ 
\begin{array}{c}
{u_{MC}}(z) \\ 
y(z)
\end{array}%
\right]+\left[ 
\begin{array}{c}
\epsilon_u(z) \\ 
\epsilon_y(z)
\end{array}%
\right]
\label{eq14-5b}
\end{align}%
where $\Pi^a$ is a stable dynamic system with $(u_{MC},y)$ as the input variables. $\epsilon_u$ and $\epsilon_y$ are $\mathcal{L}_2$-bounded, unknown, and assumed to be independent of $u,y$, which can be considered as the additive attacks. $\Pi^a(z)\left[ 
\begin{array}{c}
{u_{MC}}(z) \\ 
y(z)
\end{array}%
\right]$ is adopted to model the possible cyber-attacks which are designed and artificially designed by attackers in possession of data $(u_{MC},y)$, which can be interpreted as multiplicative attacks.
Here, without loss of generality, it is assumed that $\bar{\Pi}^{a}=I+\left[ 
\begin{array}{cc}
I & 0 \\ 
0 & 0%
\end{array}%
\right] \Pi ^{a}$ is invertible. 
\begin{theorem}
\label{th4}
Given the CPS system  (\ref{eq15-1a})-(\ref{eq14-0a}), the system dynamics under the cyber-attacks  (\ref{eq14-5b}) are governed by 
\begin{align}
&\left[ 
\begin{array}{c}
u \\ 
y%
\end{array}%
\right]  =\left[ 
\begin{array}{c}
M \\ 
N%
\end{array}%
\right] \bar{v}+\left[ 
\begin{array}{c}
-\hat{Y}+MQ \\ 
\hat{X}+NQ
\end{array}%
\right] r_{y}  \notag \\
& -\left[ 
\begin{array}{c}
M \\ 
N%
\end{array}%
\right] \left( I+\Phi \left[ 
\begin{array}{c}
M \\ 
N%
\end{array}%
\right] \right) ^{-1}\Phi \left[ 
\begin{array}{c}
M \\ 
N%
\end{array}%
\right] \bar{v}  \notag \\
& -\hspace{-2pt}\left[ 
\begin{array}{c}
M \\ 
N%
\end{array}%
\right] \hspace{-2pt}\left( \hspace{-2pt}I\hspace{-1pt}+\hspace{-1pt}\Phi \left[ 
\begin{array}{c}
M \\ 
N%
\end{array}%
\right] \right) ^{-1}\Phi \left[ 
\begin{array}{c}
-\hspace{-2pt}\hat{Y}\hspace{-2pt}+\hspace{-2pt}MQ \\ 
\hat{X}\hspace{-2pt}+\hspace{-2pt}NQ%
\end{array}%
\right] r_{y}\notag\\
&+\left[ 
\begin{array}{c}
M \\ 
N%
\end{array}%
\right] \left( I+\Phi \left[ 
\begin{array}{c}
M \\ 
N%
\end{array}%
\right] \right) ^{-1}\Psi \left[ 
\begin{array}{c}
\epsilon_u\\ 
\epsilon_y
\end{array}%
\right] \label{eq15-17} 
\end{align}%
where  
\begin{align}
&\Phi \hspace{-2pt} =\hspace{-2pt}-\hspace{-2pt}\left[ 
\begin{array}{cc}
X\hspace{-2pt}+\hspace{-2pt}Q\hat{N} & \text{ }-Y\hspace{-2pt}+\hspace{-2pt}Q\hat{M}%
\end{array}%
\right] \Pi ^{a}\hspace{-2pt}\left( \bar{\Pi}^{a}\right) ^{-1}  \notag\\
&\Psi\hspace{-2pt}=\hspace{-2pt}\left[ \hspace{-2pt}
\begin{array}{cc}
X\hspace{-2pt}+\hspace{-2pt}Q\hat{N} & -\hspace{-2pt}Y\hspace{-2pt}+\hspace{-2pt}Q\hat{M}%
\end{array}%
\hspace{-2pt}\right] \hspace{-2pt}-\hspace{-2pt}\Phi \left[ \hspace{-2pt}
\begin{array}{cc}
-I & 0 \\ 
0 & 0%
\end{array}%
\hspace{-2pt}\right]\hspace{-2pt},\bar{v}\hspace{-2pt}=\hspace{-2pt}\left( X\hspace{-2pt}+\hspace{-2pt}Q\hat{N}\right)v.\notag
\end{align}%
Moreover, the CPS dynamic can be equivalently expressed by 
\begin{align}
& \left\{ 
\begin{array}{l}
y\hspace{-2pt}=\hspace{-2pt}Gu\hspace{-2pt}+\hspace{-2pt}\hat{M}^{-1}r_{y} \\ 
u\hspace{-2pt}=\hspace{-2pt}K_{\Delta }y+\hspace{-2pt}\Delta_{\bar{v}} +d,%
\end{array}%
\right.  \label{eq15-18} \\
&\Delta _{\bar{v}} \hspace{-2pt}=\hspace{-2pt}\left( X\hspace{-2pt}+\hspace{-2pt}Q\hat{N}\hspace{-2pt}+\hspace{-2pt}\Delta K_{u}\right) ^{-1}\hspace{-2pt}\bar{v}, d\hspace{-2pt}=\hspace{-2pt}(X\hspace{-2pt}+\hspace{-2pt}Q\hat{N}\hspace{-2pt}+\hspace{-2pt}\Delta {K_u})^{-1}\Psi \hspace{-2pt}\left[ \hspace{-2pt}
\begin{array}{c}
\epsilon_u\\ \epsilon_y
\end{array}%
\hspace{-2pt}\right]\notag\\
&K_{\Delta } \hspace{-2pt}=\hspace{-2pt}-\hspace{-2pt}\left( X\hspace{-2pt}+\hspace{-2pt}Q\hat{N}\hspace{-2pt}+\hspace{-2pt}\Delta K_{u}\right) ^{-1}\left( Y\hspace{-2pt}-\hspace{-2pt}Q\hat{M}%
\hspace{-2pt}+\hspace{-2pt}\Delta K_{y}\right) , \notag\\
&\Delta {K_u}=\Phi(:,1:k_u),\Delta {K_y}\hspace{-2pt}=\hspace{-2pt}\Phi(:,k_u\hspace{-2pt}+\hspace{-2pt}1:k_u+k_y).\notag
\end{align}
\end{theorem}
\begin{proof}
It follows from the proof of Theorem \ref{th3} that the system dynamics can be described by 
\begin{align*}
&\left[  \hspace{-2pt}
\begin{array}{c}
u \\ 
y%
\end{array}%
 \hspace{-2pt}\right]  
 \hspace{-4pt} =\hspace{-4pt}\left[ \hspace{-2pt}
\begin{array}{c}
M \\ 
N%
\end{array}%
\hspace{-2pt}\right] \hspace{-4pt}\left( \hspace{-2pt}\bar{v}\hspace{-2pt}+\hspace{-4pt}\left[ \hspace{-2pt}
\begin{array}{cc}
X\hspace{-2pt}+\hspace{-2pt}Q\hat{N} & -\hspace{-2pt}Y\hspace{-2pt}+\hspace{-2pt}Q\hat{M}%
\end{array}%
\right] \hspace{-4pt}\left(\hspace{-2pt}\Pi ^{a}\hspace{-4pt}\left[ \hspace{-2pt}
\begin{array}{c}
u_{MC} \\ 
y%
\end{array}%
\hspace{-2pt}\right] \hspace{-4pt}+\hspace{-4pt}\left[ \hspace{-2pt}
\begin{array}{c}
\epsilon_u\\
\epsilon_y
\end{array}%
\hspace{-2pt}\right] \hspace{-2pt}\right)\hspace{-2pt}\right)  \\
& \;\;\;\;\;\;\;\;+\left[ \hspace{-2pt}
\begin{array}{c}
-\hat{Y}+MQ \\ 
\hat{X}+NQ%
\end{array}%
\right] r_{y}.
\end{align*}%
Observe that 
\begin{align*}
&\left[ \hspace{-2pt}
\begin{array}{c}
u_{MC} \\ 
y%
\end{array}%
\hspace{-2pt}\right] \hspace{-3pt}=\hspace{-3pt}\left[ 
\begin{array}{c}
u \\ 
y%
\end{array}%
\right] \hspace{-3pt}+\hspace{-3pt}\left[ \hspace{-2pt}
\begin{array}{cc}
-I & 0 \\ 
0 & 0%
\end{array}%
\hspace{-2pt}\right]\hspace{-2pt} \left(\Pi ^{a}\hspace{-2pt}\left[ \hspace{-2pt}
\begin{array}{c}
u_{MC} \\ 
y%
\end{array}%
\hspace{-2pt}\right]+\left[ \hspace{-2pt}
\begin{array}{c}
\epsilon_u\\
\epsilon_y
\end{array}%
\hspace{-2pt}\right]\right)\hspace{-4pt} \\
&=\hspace{-4pt}\left(\hspace{-2pt} I\hspace{-2pt}+\hspace{-2pt}\left[\hspace{-2pt} 
\begin{array}{cc}
I & 0 \\ 
0 & 0%
\end{array}%
\hspace{-2pt}\right] \hspace{-2pt}\Pi ^{a}\hspace{-2pt}\right) ^{\hspace{-2pt}-\hspace{-1pt}1}\hspace{-2pt}\left[ \hspace{-2pt}
\begin{array}{c}
u \\ 
y%
\end{array}%
\hspace{-2pt}\right] \hspace{-4pt}+\hspace{-4pt}\left(\hspace{-2pt} I\hspace{-2pt}+\hspace{-2pt}\left[ \hspace{-2pt}
\begin{array}{cc}
I & 0 \\ 
0 & 0%
\end{array}%
\hspace{-2pt}\right] \hspace{-3pt}\Pi ^{a}\hspace{-2pt}\right) ^{\hspace{-2pt}-\hspace{-1pt}1}\hspace{-2pt}\left[ \hspace{-4pt}
\begin{array}{cc}
-\hspace{-1pt}I & 0 \\ 
0 & 0%
\end{array}%
\hspace{-2pt}\right]\hspace{-3pt}\left[ \hspace{-2pt}
\begin{array}{c}
\epsilon_u\\ \epsilon_y
\end{array}%
\hspace{-2pt}\right]
\end{align*}%
which yields%
\begin{align*}
\left[ \hspace{-2pt}
\begin{array}{c}
u \\ 
y%
\end{array}%
\hspace{-2pt}\right]  \hspace{-2pt}&=\hspace{-2pt}\left[ \hspace{-2pt}
\begin{array}{c}
M \\ 
N%
\end{array}%
\hspace{-2pt}\right]\hspace{-2pt} \left( \hspace{-2pt}\bar{v}\hspace{-2pt}-\hspace{-2pt}\Phi \hspace{-2pt}\left[ 
\begin{array}{c}
u \\ 
y%
\end{array}%
\right] \hspace{-2pt}+\hspace{-2pt}\Psi\hspace{-2pt}\left[ 
\begin{array}{c}
\epsilon_u\\ \epsilon_y
\end{array}%
\right] \hspace{-2pt}\right) \hspace{-3pt}+\hspace{-2pt}\left[ \hspace{-2pt}
\begin{array}{c}
-\hspace{-2pt}\hat{Y}\hspace{-2pt}+\hspace{-2pt}MQ \\ 
\hat{X}\hspace{-2pt}+\hspace{-2pt}NQ%
\end{array}%
\hspace{-2pt}\right]\hspace{-2pt} r_{y}
  \\
&=\hspace{-2pt}\left( I\hspace{-2pt}+\hspace{-2pt}\left[ \hspace{-2pt}
\begin{array}{c}
M \\ 
N%
\end{array}%
\hspace{-2pt}\right] \Phi \right) ^{-1}\hspace{-4pt}\left( \hspace{-2pt}\left[ \hspace{-2pt}
\begin{array}{c}
M \\ 
N%
\end{array}%
\hspace{-2pt}\right]\hspace{-4pt} \left(\hspace{-2pt}\bar{v}\hspace{-2pt}+\hspace{-2pt}\Psi\hspace{-2pt}\left[ 
\begin{array}{c}
\epsilon_u\\ \epsilon_y
\end{array}%
\right] \hspace{-2pt} \right)\hspace{-3pt}+\hspace{-3pt}\left[ \hspace{-2pt}
\begin{array}{c}
-\hspace{-2pt}\hat{Y}\hspace{-2pt}+\hspace{-2pt}MQ \\ 
\hat{X}\hspace{-2pt}+\hspace{-2pt}NQ%
\end{array}%
\hspace{-2pt}\right]\hspace{-2pt} r_{y}\hspace{-2pt}\right) \hspace{-2pt}.
\end{align*}%
After some routine calculations using the relation 
\begin{equation*}
\setlength{\abovedisplayskip}{3pt}\setlength{\belowdisplayskip}{3pt}
\left( I\hspace{-2pt}+\hspace{-2pt}\left[ 
\begin{array}{c}
M \\ 
N%
\end{array}%
\right]\hspace{-2pt} \Phi \right) ^{-1}\hspace{-2pt}=\hspace{-2pt}I\hspace{-2pt}-\hspace{-2pt}\left[ 
\begin{array}{c}
M \\ 
N%
\end{array}%
\right] \left( I\hspace{-2pt}+\hspace{-2pt}\Phi \left[ 
\begin{array}{c}
M \\ 
N%
\end{array}%
\right] \right) ^{-1}\hspace{-2pt}\Phi ,
\end{equation*}%
we have (\ref{eq15-17}). 
Noticing that
\begin{align*}
&u=Ky^a+v+a_{u_{MC}}=Ky+Ka_y+v+a_{u_{MC}}\Longrightarrow\\
&(X\hspace{-2pt}+\hspace{-2pt}Q\hat{N})u\hspace{-2pt}=\hspace{-2pt}(-\hspace{-2pt}Y\hspace{-2pt}+\hspace{-2pt}Q\hat{M})(y\hspace{-2pt}+\hspace{-2pt}a_y)\hspace{-2pt}+\hspace{-2pt}(X\hspace{-2pt}+\hspace{-2pt}Q\hat{N})(v\hspace{-2pt}+\hspace{-2pt}a_{u_{MC}})\\
&\;\;\;\;=\hspace{-2pt}(-\hspace{-1pt}Y\hspace{-2pt}+\hspace{-2pt}Q\hat{M})y\hspace{-2pt}+\hspace{-2pt}(X\hspace{-2pt}+\hspace{-2pt}Q\hat{N})v\hspace{-2pt}-\hspace{-2pt}\Phi \left[ 
\begin{array}{c}
u\\y
\end{array}%
\right] \hspace{-2pt}+\hspace{-2pt}\Psi\left[ 
\begin{array}{c}
\epsilon_u\\ \epsilon_y
\end{array}%
\right],
\end{align*}
 (\ref{eq15-18}) is finally proved. 
\end{proof}

It follows from (\ref{eq15-17}) in Theorem \ref{th4}  that the system dynamics consists of five terms. The first two terms give the nominal system response (attack-free dynamic), while the third and fourth terms describe the deviations caused by the cyber-attacks modelled by (%
\ref{eq14-5b}) in the responses to the reference signal and to the plant residual, respectively. The fifth term gives the deviations caused by the additive cyber-attacks $\epsilon_u,\epsilon_y$. 

When $\epsilon_u=0,\epsilon_y=0$, $a_y,a_{u_{MC}}$ can be interpreted as the multiplicative cyber-attacks. The equivalent configuration of the CPS closed-loop dynamic under multiplicative cyber-attacks is sketched in Fig. \ref{fig3}. It is evident that the uncertainty in the feedback controller in case of the multiplicative
attacks,  may affect the loop stability and feedback performance.


\begin{remark}
It is worth mentioning that
the additive cyber-attacks (\ref{eq14-4b}) deviate the system response to the reference signal from the nominal value, while the multiplicative cyber-attacks (\ref{eq14-5b}) lead to the change both to the reference response and feedback control performance.
\end{remark}

\begin{figure}[!t]
\centering\includegraphics[width=5.5cm]{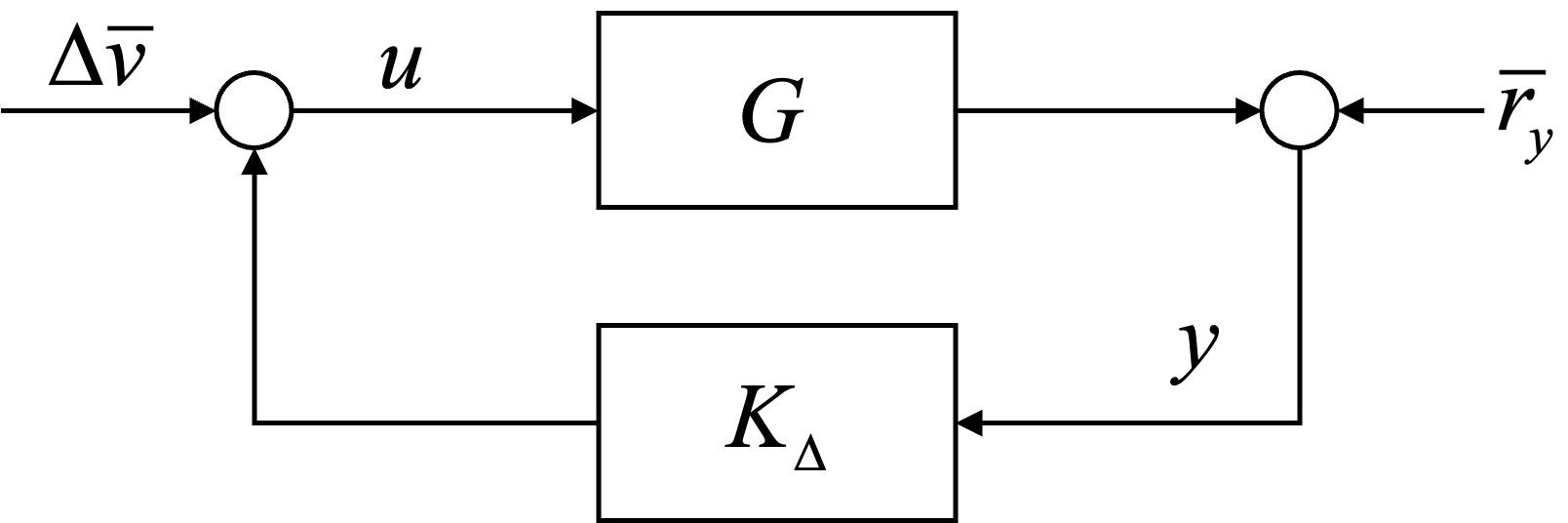}
\caption{The equivalent system dynamics for CPS systems with multiplicative cyber-attacks}
\label{fig3}
\end{figure}

\subsection{Analysis of system dynamics under cyber-attacks and faults}
On the basis of  Theorems \ref{th3}-\ref{th4}
and the Bezout identity 
\begin{align}
-\hat{N}M+\hat{M}N=0
\end{align}
it is evident that $r_y$ is independent of $a_y,a_{u_{MC}}$.
Thus, the following corollary is straightforward.
\begin{corollary}
Given CPS modelled by (\ref{eq15-1a})-(\ref{eq14-0a}),
both the additive attacks (\ref{eq14-4b}) and multiplicative attacks (\ref{eq14-5b}), lead to variations exclusively in the system image subspace $\mathcal{I}_G$, while the process uncertainties and faults, reflected by the residual $r_y$, results in change solely in the controller image  subspace $\mathcal{I}_K$.
\end{corollary}
It is apparent that the residual generator, 
$
r_{y}=\hat{M}y-\hat{N}u,
$
embedded and implemented at the plant side  
only reflects the uncertainties and faults in the residual subspace, and thus can not detect the cyber-attacks.
 In addition, the following claims are of considerable importance:


\begin{itemize}
\item detection and control of CPSs under cyber-attacks and dynamic
processes suffering uncertainties or faults are two different but dual
problems. This conclusion illustrates that the common view of considering
distinguishing/isolating cyber-attacks and process faults as a challenging
issue, because of their similar forms, is less convincing,

\item instead, the real challenges are the development of novel detection
and (feedback) control schemes that are capable to detect and control the
uncertainties in the plant image subspace. It is known that the traditional detection and control
methods are dedicated to handling  the uncertainties and faults in the plant.
Since those uncertainties, including unknown inputs and faults, are indeed
uncertainties in the controller image subspace, which is an 
complement of the system image subspace, the traditional methods are less (or even
not) capable to detect and control of cyber-attacks, and

\item concretely, methods should be developed to detect and control
of cyber-attacks in the plant image subspace. 
\end{itemize}
\noindent The above observations motivate us to find an alternative solution for the integrated design of detection  and resilient control  of both technical faults and cyber-attacks for CPSs in the subsequent sections.

\section{Alternative Representation of Process Signal}
The fact that  the output residual $r_{y}$ cannot detect and control cyber-attacks  implies
that additional information about the CPS dynamics is needed. 
Intuitively,
informations related to  the input residual is an
immediate candidate. 
This motivates us to define the  input residual $r_u$ on the plant side  as
\begin{equation}
r_u(z)\hspace{-2pt}:=\hspace{-2pt}(X(z)\hspace{-2pt}+\hspace{-2pt}Q(z)\hat{M}(z))u(z)\hspace{-2pt}+\hspace{-2pt}(Y(z)\hspace{-2pt}-\hspace{-2pt}Q(z)\hat{M}(z))y(z)\label{eq-ru-generate}
\end{equation}
Substituting (\ref{eq14-10}) into (\ref{eq-ru-generate}) leads to 
\begin{equation}
\setlength{\abovedisplayskip}{3pt}\setlength{\belowdisplayskip}{3pt}
r_u(z)\hspace{-2pt}:=\hspace{-2pt}\bar{v}+\hspace{-2pt}\left( \hspace{-1pt}X\hspace{-2pt}+\hspace{-2pt}Q\hat{N}\hspace{-1pt}\right) a_{u_{MC}}\hspace{-2pt}-\hspace{-2pt}\left( \hspace{-2pt}Y\hspace{-2pt}-\hspace{-2pt}Q\hat{M}\hspace{-2pt}\right)\hspace{-2pt}a_y.
\end{equation}
It is evident that $r_u$ is a natural information provider of the cyber-attacks $\left(
a_{y},a_{u_{CM}}\right)$. 
In this section,  the  role of $r_u$  in the closed-loop dynamics is further examined,
 which lays the foundation for the subsequent study. 


\subsection{The parameterization of closed-loop dynamics}
Without loss of generality, we consider the process (\ref{eq14-0a}) with the controller $u(z)=-X^{-1}(z)Y(z)y(z)+v(z)$. It follows from the previous discussions that the process can be equivalently described by
\begin{align}
\begin{cases}\hat{x}(k\hspace{-2pt}+\hspace{-2pt}1)\hspace{-2pt}=\hspace{-3pt}(A+BF)\hat{x}(k)\hspace{-2pt}+\hspace{-2pt}Br_u(k)\hspace{-2pt}+\hspace{-2pt}Lr_y(k)\\
y(k)=C\hat{x}(k)+Du(k)+r_{y}(k)\\
u(k)=F\hat{x}(k)+\bar{v}(k),\bar{v}(z)=X(z)v(z)\\
r_y(k)=y(k)-C\hat{x}(k)-Du(k)\label{eq-observera0}
\end{cases}
\end{align}
where
\begin{equation}
\setlength{\abovedisplayskip}{3pt}\setlength{\belowdisplayskip}{3pt}
r_u(z):=X(z)u(z)+Y(z)y(z)=u(z)-F\hat{x}(z).
\end{equation}
It is noteworthy that the selection of $F,L$ has no influence on the transfer function $G(z)$. That is, $G(z)$ is invariant to $F,L$.  Consequently, we have the following theorem.

\begin{theorem}\label{th1}
Given $F_i\neq F,L_i\neq L,i=1,\cdots,\kappa$ which ensures that $A_{F_i}=A+BF_i$ is Schure, (\ref{eq-observera0}) can be equivalently realized by 
\begin{align}\hspace{-2pt}
\begin{cases}\hspace{-2pt}
\hat{x}(k\hspace{-2pt}+\hspace{-2pt}1)\hspace{-2pt}=\hspace{-3pt}A_{F_i}\hat{x}(k)\hspace{-2pt}+\hspace{-2pt}Br_{u,i}(k)\hspace{-2pt}+\hspace{-2pt}L_ir_{y,i}(k)\notag\\
\hspace{-2pt}u(k)\hspace{-2pt}=\hspace{-2pt}F_i\hat{x}(k)\hspace{-2pt}+\hspace{-2pt}r_{Q_i}\hspace{-2pt}(k)\hspace{-2pt}+\hspace{-2pt}\bar{v}_i(k)\hspace{-1pt},\bar{v}_i(z)\hspace{-3pt}=\hspace{-3pt}(X_i(z)\hspace{-3pt}+\hspace{-2pt}Q_i(z)\hat{N}_i(z))v(z)\notag\\
\hspace{-2pt}r_{Q_i}(z)\hspace{-2pt}=\hspace{-2pt}M_i(z)Q_i(z)r_{y,i}(z)\notag\\
\hspace{-2pt}y(k)=C\hat{x}(k)+Du(k)+r_{y,i}(k)\notag
\end{cases}
\label{eq-observera}
\end{align}
where  $X_i,Y_i,M_i,N_i,\hat{X}_i,\hat{Y}_i,\hat{N}_i,\hat{M}_i$ are transfer functions given by (\ref{eq-TF-MN})
with $F=F_i,L=L_i$, and  
\begin{align*}
&r_{u,i}(z):=X_i(z)u(z)+Y_i(z)y(z)=u(z)-F_i\hat{x}(z)\\
&r_{y,i}(z)=-\hat{N}_i(z)u(z)\hspace{-2pt}+\hspace{-2pt}\hat{M}_i(z)y(z)\hspace{-2pt}=\hspace{-2pt}y(z)\hspace{-2pt}-\hspace{-2pt}C\hat{x}(z)\hspace{-2pt}-\hspace{-2pt}Du(z)\\
&\bar{V}_{i0}=\left[ 
\begin{array}{cc}
X & \text{ }Y%
\end{array}%
\right] \left[ 
\begin{array}{c}
-\hat{Y}_i\\ 
\hat{X}_i
\end{array}%
\right] ,\bar{R}_{i0}=\left[ 
\begin{array}{cc}
X_i & \text{ }Y_i%
\end{array}%
\right] \left[ 
\begin{array}{c}
-\hat{Y} \\ 
\hat{X}
\end{array}%
\right] \\
&V_{i0}=I+(F_i-F)(zI-A_{F_i})^{-1}B,Q_i(z)=\bar{R}_{i0}(z)R_{0i}(z)\\
&R_{0i}=I-C(zI-A_{L})^{-1}(L-L_i),r_{Q_i}(z)=Q_i(z)r_{y,i}(z).
\end{align*}
\end{theorem}

\begin{proof}
It is straightforward that
\begin{align}
u(k)=F\hat{x}+\bar{v}=-\hat{Y}r_y+M\bar{v}.
\end{align}
By means of Bezout identity (\ref{eq-3a}), $\left[ 
\begin{array}{cc}
X & \text{ }Y%
\end{array}%
\right]$ and $\left[ 
\begin{array}{c}
-\hat{Y} \\ 
\hat{X}%
\end{array}%
\right]$ can be parameterized respectively by 
\begin{align*}
\left[ \hspace{-2pt}
\begin{array}{c}
-\hspace{-2pt}\hat{Y} \\ 
\hat{X}%
\end{array}%
\hspace{-2pt}\right]\hspace{-2pt}=\hspace{-2pt}\left[ \hspace{-2pt}
\begin{array}{cc}
M_i & -\hat{Y}_i\\ 
N_i & \hat{X}_i
\end{array}%
\hspace{-2pt}\right]\hspace{-2pt}\left[ \hspace{-2pt}
\begin{array}{c}
Q_{1i} \\ 
Q_{2i} 
\end{array}%
\hspace{-2pt}\right]\hspace{-2pt},\left[ \hspace{-2pt}
\begin{array}{cc}
X & Y%
\end{array}%
\hspace{-2pt}\right]\hspace{-2pt}=\hspace{-2pt}\left[ \hspace{-2pt}
\begin{array}{cc}
\bar{Q}_{1i} & \bar{Q}_{2i} 
\end{array}%
\hspace{-2pt}\right]\hspace{-2pt}\left[ 
\begin{array}{cc}
X_i & Y_i\\
-\hspace{-2pt}\hat{N}_i & \hat{M}_i
\end{array}%
\hspace{-2pt}\right]
\end{align*}
for some $Q_{1i},Q_{2i},\bar{Q}_{1i} ,\bar{Q}_{2i} \in \mathcal{RH}_\infty$. 
It follows from\cite{Ding2020} that
\begin{align*}
\left[ 
\begin{array}{cc}
-\hspace{-2pt}\hat{N}_i & \hat{M}_i
\end{array}%
\hspace{-2pt}\right]=R_{i0}\left[ 
\begin{array}{cc}
-\hspace{-2pt}\hat{N} & \hat{M}
\end{array}%
\hspace{-2pt}\right],\left[ \hspace{-2pt}
\begin{array}{c}
M_i \\ 
N_i 
\end{array}%
\hspace{-2pt}\right]=\left[ \hspace{-2pt}
\begin{array}{c}
M \\ 
N
\end{array}%
\hspace{-2pt}\right]V_{i0}
\end{align*}
and $R_{i0}^{-1}=R_{0i}, V_{i0}^{-1}=V_{0i}$. Consequently, we have
\begin{align*}
&\left[ \hspace{-2pt}
\begin{array}{cc}
-\hspace{-2pt}\hat{N} & \hat{M}
\end{array}%
\hspace{-2pt}\right]\hspace{-4pt}\left[ \hspace{-2pt}
\begin{array}{c}
-\hspace{-2pt}\hat{Y} \\ 
\hat{X}%
\end{array}%
\hspace{-2pt}\right]\hspace{-4pt}=\hspace{-2pt}R_{0i}\hspace{-2pt}\left[ \hspace{-2pt}
\begin{array}{cc}
-\hspace{-2pt}\hat{N}_i & \hat{M}_i
\end{array}%
\hspace{-2pt}\right]\hspace{-3pt}\left[ \hspace{-2pt}
\begin{array}{cc}
M_i & -\hat{Y}_i\\ 
N_i & \hat{X}_i
\end{array}%
\hspace{-2pt}\right]\hspace{-2pt}\left[ \hspace{-2pt}
\begin{array}{c}
Q_{1i} \\ 
Q_{2i} 
\end{array}%
\hspace{-2pt}\right]\hspace{-4pt}=\hspace{-2pt}R_{0i}Q_{2i} \hspace{-2pt}=\hspace{-2pt}I\\
&\left[ \hspace{-2pt}
\begin{array}{cc}
X & Y%
\end{array}%
\hspace{-2pt}\right]\hspace{-2pt}\left[ \hspace{-2pt}
\begin{array}{c}
M \\ 
N
\end{array}%
\hspace{-2pt}\right]\hspace{-2pt}=\hspace{-2pt}\left[ \hspace{-2pt}
\begin{array}{cc}
\bar{Q}_{1i} & \bar{Q}_{2i} 
\end{array}%
\hspace{-2pt}\right]\hspace{-2pt}\left[ 
\begin{array}{cc}
X_i & Y_i\\
-\hspace{-2pt}\hat{N}_i & \hat{M}_i
\end{array}%
\hspace{-2pt}\right]\hspace{-2pt}\left[ \hspace{-2pt}
\begin{array}{c}
M_i \\ 
N_i 
\end{array}%
\hspace{-2pt}\right]\hspace{-2pt}V_{0i}\hspace{-2pt}=\hspace{-2pt}\bar{Q}_{1i}V_{0i}\hspace{-2pt}=\hspace{-2pt}I
\end{align*}
Moreover, it is evident that
\begin{align*}
&V_{i0}=\left[ 
\begin{array}{cc}
X & Y
\hspace{-2pt}\end{array}%
\hspace{-2pt}\right]\hspace{-2pt}\left[ \hspace{-2pt}
\begin{array}{c}
-\hspace{-2pt}\hat{Y}_i \\ 
\hat{X}_i%
\end{array}%
\hspace{-2pt}\right]=\left[ \hspace{-2pt}
\begin{array}{cc}
\bar{Q}_{1i} & \bar{Q}_{2i} 
\end{array}%
\hspace{-2pt}\right]\hspace{-2pt}\left[ 
\begin{array}{cc}
X_i & Y_i\\
-\hspace{-2pt}\hat{N}_i & \hat{M}_i
\end{array}%
\hspace{-2pt}\right]\hspace{-2pt}\left[ \hspace{-2pt}
\begin{array}{c}
-\hspace{-2pt}\hat{Y}_i \\ 
\hat{X}_i%
\end{array}%
\hspace{-2pt}\right]=\bar{Q}_{2i} \\
&\bar{R}_{i0}=\left[ 
\begin{array}{cc}
X_i & Y_i
\hspace{-2pt}\end{array}%
\hspace{-2pt}\right]\hspace{-2pt}\left[ \hspace{-2pt}
\begin{array}{c}
-\hspace{-2pt}\hat{Y} \\ 
\hat{X}%
\end{array}%
\hspace{-2pt}\right]=\left[ 
\begin{array}{cc}
X_i & Y_i
\hspace{-2pt}\end{array}%
\hspace{-2pt}\right]\hspace{-2pt}\left[ \hspace{-2pt}
\begin{array}{cc}
M_i & -\hat{Y}_i\\ 
N_i & \hat{X}_i
\end{array}%
\hspace{-2pt}\right]\hspace{-2pt}\left[ \hspace{-2pt}
\begin{array}{c}
Q_{1i} \\ 
Q_{2i} 
\end{array}%
\hspace{-2pt}\right]=Q_{1i}.
\end{align*}
As a result, one has that
\begin{align*}
u&=-\hat{Y}_iR_{i0}r_y+M_i\bar{R}_{i0}r_y+M_iV_{i0}(V_{0i}X_i-\bar{V}_{i0}\hat{N}_i)v\\
&=-\hat{Y}_ir_{y,i}+M_i\bar{R}_{i0}R_{0i}r_{y,i}+M_i(X_i-V_{0i}\bar{V}_{i0}\hat{N}_i)v\\
&=F_i\hat{x}+r_{Q_i}+\bar{v}_i
\end{align*}
which completes the proof.
\end{proof}
\begin{remark}
Theorem \ref{th1} reveals that  varying $F,L$ to $F_i, L_i$ is equivalent to adding a stable post-filter and additional residual signal  to $u,v$ without changing the process dynamics.  \end{remark}


 The closed-loop dynamics for different setting of $L_i,F_i$ is analyzed in the following corollary.
 \begin{corollary}\label{th2}
The closed-loop dynamics for process (\ref{eq-observera0}) can be equivalently described by
\begin{align}
&\left[ \hspace{-2pt}
\begin{array}{c}
u \\ 
y%
\end{array}%
\hspace{-2pt}\right] \hspace{-2pt}=\hspace{-2pt}\left[ \hspace{-2pt}
\begin{array}{c}
M \\ 
N
\end{array}%
\hspace{-2pt}\right] r_u\hspace{-2pt}+\hspace{-4pt}\left[ \hspace{-2pt}
\begin{array}{cc}
 -\hspace{-2pt}\hat{Y} \\ 
 \hat{X}
\end{array}%
\hspace{-2pt}\right]r_y\hspace{-2pt} = \hspace{-2pt}\left[ \hspace{-2pt}
\begin{array}{c}
M \\ 
N%
\end{array}%
\hspace{-2pt}\right]\hspace{-2pt}Q_{u_i} r_{u,i}\hspace{-2pt}+\hspace{-4pt}\left[ \hspace{-2pt}
\begin{array}{c}
-\hspace{-2pt}\hat{Y}\hspace{-2pt}+\hspace{-2pt}MQ_{e_i} \\ 
\hat{X}\hspace{-2pt}+\hspace{-2pt}NQ_{e_i}%
\end{array}%
\hspace{-2pt}\right] \hspace{-2pt}r_{y,i}\notag\\
&Q_{u_i}=V_{i0},Q_{e_i}=\bar{V}_{i0}+\bar{R}_{i0}R_{i0}^{-1}.
\label{eq-ruya}
\end{align}
 \end{corollary}
 \begin{proof}
%
It follows from the poof of Theorem \ref{th1} that the system dynamics (\ref{eq-observera0}) can be equivalently written as 
\begin{align*}
&\left[ 
\begin{array}{c}
u \\ 
y%
\end{array}%
\right] \hspace{-2pt}=\left[ 
\begin{array}{cc}
X& \text{ }Y\\ 
-\hat{N} & \text{ }\hat{M}%
\end{array}%
\right]^{-1}\left[ 
\begin{array}{c}
r_u \\ 
r_y
\end{array}%
\right] = \hspace{-2pt}\left[ 
\begin{array}{c}
M \\ 
N 
\end{array}%
\right]r_u+\left[ 
\begin{array}{cc}
 -\hat{Y} \\ 
 \hat{X}
\end{array}%
\right]r_y\\
&=\left[ 
\begin{array}{c}
M \\ 
N
\end{array}%
\right]\hspace{-2pt}\left(V_{i0}r_{u,i}\hspace{-2pt}+\hspace{-2pt}\bar{V}_{i0}r_{y,i}\right)\hspace{-2pt}+\hspace{-2pt}\left(\left[ 
\begin{array}{c}
-\hspace{-2pt}\hat{Y}_{i} \\ 
\hat{X}_{i}%
\end{array}%
\right]R_{i0}\hspace{-2pt}+\hspace{-2pt}\left[ 
\begin{array}{c}
M_{1} \\ 
N_{1}%
\end{array}%
\right]\bar{R}_{i0} \right)\hspace{-2pt}r_y.
\end{align*}
On the basis of \cite{Ding2014}, we have
\begin{align*}
R_{i0}\left[ 
\begin{array}{cc}
-\hat{N}& \text{ }\hat{M}
\end{array}%
\right]=\left[ 
\begin{array}{cc}
-\hat{N}_i& \text{ }\hat{M}_i
\end{array}%
\right]\Longrightarrow R_{i0}r_y=r_{y,i}
\end{align*}
which completes the proof.\end{proof}

\begin{remark}
(\ref{eq-ruya}) can be considered as a full parameterization of CPS closed-loop dynamic with $Q_{u_i},Q_{e_i}\in \mathcal{RH}_\infty$ as the parameter systems.
\end{remark}

\subsection{The process signal subspace and residual subspaces}
It follows from the proof of Corollary \ref{th2} that
 \begin{equation}
 \setlength{\abovedisplayskip}{5pt}\setlength{\belowdisplayskip}{5pt}
\left[ 
\begin{array}{c}
r_u\\ 
r_{y}%
\end{array}%
\right]=\left[ 
\begin{array}{cc}
X& \text{ }Y\\ 
-\hat{N} & \text{ }\hat{M}%
\end{array}%
\right]\left[ 
\begin{array}{c}
u \\ 
y%
\end{array}%
\right]. \label{eq-aa1}\end{equation}
It is evident that the
process data $\left( u,y\right) $ and the I/O residual signals $\left(
r_{u},r_{y}\right) $ is one-to-one mapping. 
Notice that $r_u,r_y$ 
build the system and controller  image subspaces $\mathcal{I}_G$ and $\mathcal{I}_K$, respectively, which are complementary. 
 Consequently, we have the following lemma.
 \begin{lemma}
 The system signal space 
 \begin{equation}
 \setlength{\abovedisplayskip}{5pt}\setlength{\belowdisplayskip}{5pt}
 \mathcal{S}=\left\{\left[ 
\begin{array}{c}
u\\y
\end{array}%
\right]\in \mathcal{H}_2^{k_u+k_y}\right\}
 \end{equation}
with $u,y$ as the input and output of the process (\ref{eq2-b}),
  can be constructed by 
 \begin{equation}
 \setlength{\abovedisplayskip}{5pt}\setlength{\belowdisplayskip}{5pt}
\mathcal{S}=\mathcal{I}_G \oplus \mathcal{I}_K.
\end{equation}
\end{lemma}
\begin{remark}
To summarize, the system I/O signal  consists of the process image and controller image subspaces. 
The image subspace of the process is driven by $r_u$, which describes the response of I/O signal to the reference and cyber-attacks, and the image subspace of the controller is composed of the I/O data as the responses to the uncertainties/faults in the plant induced by $r_y$.
\end{remark}
Recall that $r_u$ delivers the information for cyber-attacks which can be applied for both attack detection and resilient control, while $r_y$ provides the information for process faults for fault detection and fault-tolerant control. 
The one-to-one mapping between the process data $(u,y)$ and the residual signals $(r_u,r_y)$, and the role of  $r_u,r_y$ in the parameterization of the closed-loop dynamics naturally triggers our inquisitiveness for alternative solutions and modifications on system configuration in the subsequent section.

\section{A Modified CPS Configuration and the Integrated Design Scheme}\label{sec4}

In this section,  a modified CPS configuration is developed. It is followed by the integrated design involving fault detection,  attack detection,  resilient and fault-tolerant control.


\subsection{A modified CPS configuration}
Before proceeding further, the main requirements for the  CPS configuration are summarized first:
i) achieving ``fail-safe" cyber-security,
ii) detecting the process faults and cyber-attacks,
iii) guaranteeing high fault-tolerance against process faults and high resilience of CPSs under cyber-attacks,
iv) ensuring the data privacy, and
v) saving online computation and communication effort.

To fulfill the above requirements,  we consider the CPS system, which consists of two main parts, the plant
with the embedded sensors and actuators as well as a computing unit,  and the monitoring and control (MC) station, at which  MC algorithms are
implemented. The subsequent CPS configuration (as shown in Fig. \ref{fig1}) is proposed:

\begin{itemize}
\item on the plant side, the received signal transmitted by the MC-station
is $u_{MC}^{a},$%
\begin{equation}
u_{MC}^{a}=u_{MC}+a_{u_{MC}},\label{eq15-38a}
\end{equation}%
the performed (online) computation is%
\begin{equation}
 \setlength{\abovedisplayskip}{5pt}\setlength{\belowdisplayskip}{5pt}
\left\{ \hspace{-2pt}
\begin{array}{l}
\hat{x}(k\hspace{-2pt}+\hspace{-2pt}1)\hspace{-2pt}=\hspace{-2pt}\left(\hspace{-2pt} A\hspace{-2pt}-\hspace{-2pt}LC\hspace{-1pt}\right)\hspace{-2pt} \hat{x}(k)\hspace{-2pt}+\hspace{-2pt}\left(\hspace{-2pt} B\hspace{-2pt}-\hspace{-2pt}LD\hspace{-1pt}\right) \hspace{-2pt}u(k)\hspace{-2pt}+\hspace{-2pt}Ly(k) \\ 
r_{u}(k)=u(k)-F\hat{x}(k)-v_0(k) ,v_0(z)=Q_v(z)\bar{v}_0(z)\\ 
r_{y}(k)\hspace{-2pt}=\hspace{-2pt} y(k)\hspace{-2pt}-\hspace{-2pt}\hat{y}(k)\hspace{-2pt}=\hspace{-2pt} y(k)\hspace{-2pt}-\hspace{-2pt}C\hat{x}%
(k)\hspace{-2pt}-\hspace{-2pt}Du(k) \\ 
u(z)=F\hat{x}(z)+v_0(z)+u_{MC}^{a}(z)\\
 r_{y,u}(z)=Q_{r,1}(z)r_y(z)+Q_{r,2}(z)r_u(z)
\end{array}%
\right.  \label{eq15-38}
\end{equation}%
where $Q_{r,1}, Q_{r,2}\hspace{-4pt}\in\hspace{-4pt} \mathcal{RH}_\infty$,  $Q_v$ is a pre-filter, $\bar{v}_0$ is the baseline for the reference, and
 $ r_{y,u} $ is sent to the MC-station,

\item at the MC-station, the received signals are 
\begin{equation*}
 \setlength{\abovedisplayskip}{5pt}\setlength{\belowdisplayskip}{5pt}
r_{y,u}^a=r_{y,u}+a_{r_{y,u}},
\end{equation*}%
and the following  online computations are performed with $u_{MC}$ being sent to the plant

\begin{itemize}
\item for the control purpose%
\begin{align}
&u_{MC}(z)\hspace{-2pt}=\hspace{-2pt}Q_{u_{MC}}(z)\hspace{-2pt}\left(\hspace{-1pt} r_{y,u}^{a}(z)\hspace{-2pt}-\hspace{-2pt}Q_{r,2}(z)v(z)\right)\hspace{-2pt}+\hspace{-2pt}v(z)\notag\\
&v(z)=Q_v(z)(\bar{v}(z)-\bar{v}_0(z), Q_v\in \mathcal{RH}_\infty\label{eq15-38c}
\end{align}
where $Q_{u_{MC}}\in \mathcal{RH}_\infty$ is the design parameter, 
$\bar{v}$ is the target reference,

\item for attack detection 
 \begin{equation}
 \setlength{\abovedisplayskip}{2pt}\setlength{\belowdisplayskip}{2pt}
\bar{r}_{\eta_a}(z)=\bar{R}(z)(r_{y,u}^a(z)-Q_{r,2}(z)v(z))\label{eq-ad-post-filter}
\end{equation}
where $\bar{R}(z)$ is the post-filter to be designed.
\end{itemize}
\end{itemize}

\begin{remark}
 The crucial arguments for the above CPS configuration are
\begin{itemize}
\item ``fail-safe" cyber-security: With the embedded controller $F\hat{x}(k)$ in (\ref%
{eq15-38}) at the plant side, the system operation is assured for the case that the
communications fail,
\item with the embedded $v_0$  as the baseline of the reference, the transmission of the full information of the target reference $v$ is avoided which increases the data security,

\item with $r_{y,u} $ sufficient informations for detection and control are
available at the MC-station, and

\item the transmission of  $r_{y,u}$
instead of the process data $y,$ increases the cyber-security with respect
to the data privacy and meanwhile saves the communication effort.
\end{itemize}
\end{remark}

It is noteworthy that a model aware adversary can design attacks by injecting false data which causes damage to the control performance based on the identified process model. To handle this issue, the moving target defense strategies have been widely applied by introducing the time-varying parameters in the control systems \cite{MT-method-IEEE-TAC2020}.
This motivates us to adopt (\ref{eq-observera}) as a time-varying model of the process with $A_{F_{i_k}}, L_{i_k},F_{i_k}$ belong to a finite set of modes 
 \begin{equation}
 \setlength{\abovedisplayskip}{5pt}\setlength{\belowdisplayskip}{5pt}
\Gamma=\{ (A_{F_1}, L_1, F_1),\cdots,(A_{F_{\kappa}}, L_{\kappa}, F_{\kappa})\}.
\end{equation}
Here, $i_k$ represents the index of the process model at the $k$-th sample. 
 In this sense, $r_{y,u,i}$ with $F_i,L_i$ is generated from the plant side, and then $r_{y,u}$ is computed  at the MC-station based on the received data $r_{y,u,i}^a=r_{y,u,i}+a_{r_{y,u}}$ and Theorem \ref{th1}.
Even though $\Gamma$ may be available to the attacker, the switching rule will be unknown.
As a result, the information available to the defender and the attacker  at time $k$ are  
\begin{equation}
\begin{aligned}
 \setlength{\abovedisplayskip}{5pt}\setlength{\belowdisplayskip}{5pt}
&\Gamma_k^d=\{ A,B,C,D,F,L,A_{F_{i_k}}, L_{i_k},F_{i_k},u_{MC}^a,y,r_{u,y,i_k})\}\notag
\\
&\Gamma_k^a=\{ A,B,C,D,F,L,u_{MC},r_{u,y,i_k}^a)\}\notag
\end{aligned}
\end{equation}
respectively. 
A significant advantage of this moving target scheme over the existing methods lies in adopting the time-varying parameters $L_i,F_i$ without changing the process dynamics.
To ease the presentation of the design scheme of the CPS system, this issue will not be discussed in details.

\begin{figure}[!t]
\centering\includegraphics[width=6.5cm]{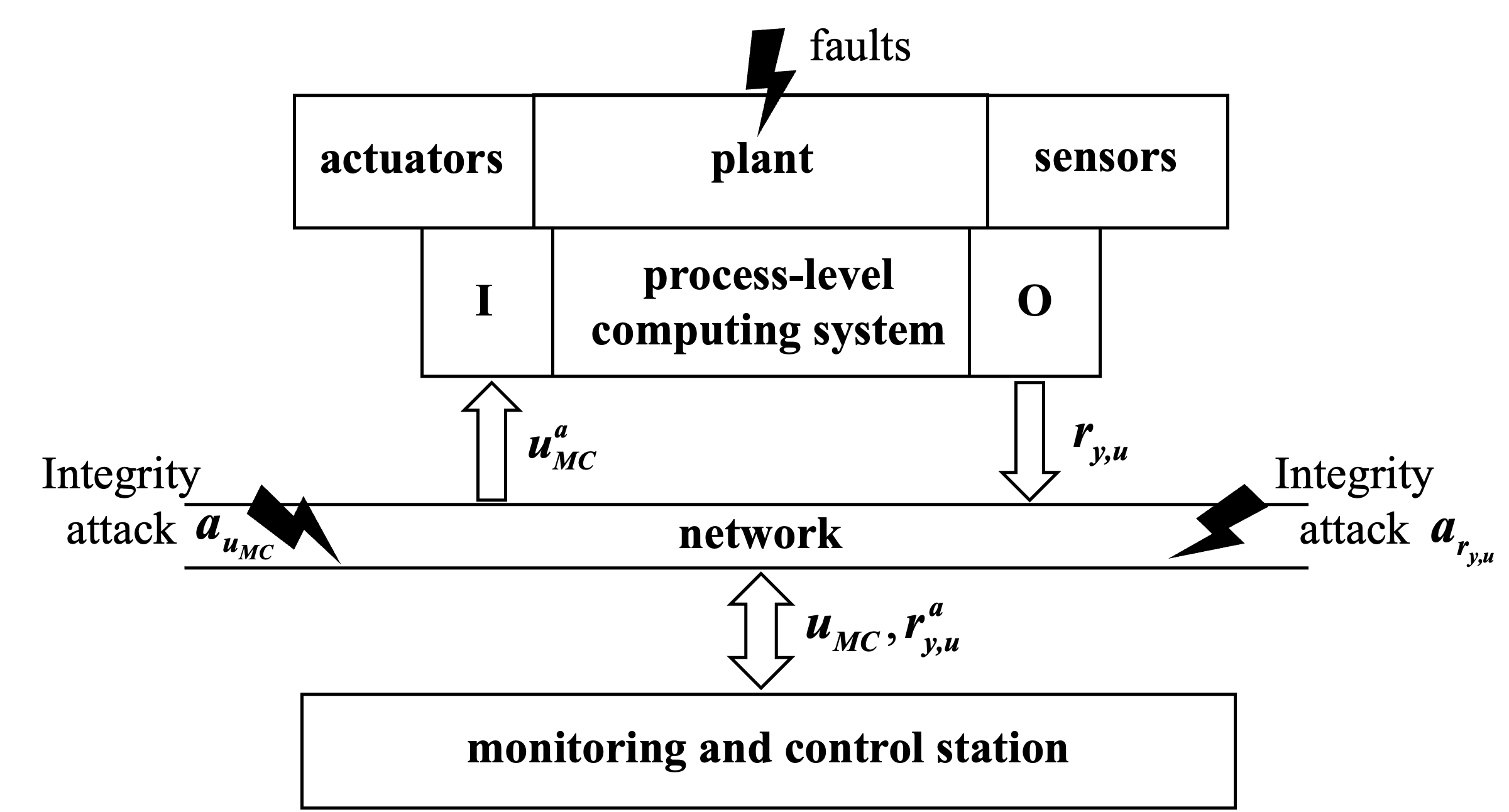}
\caption{A modified CPS configuration}
\label{fig1}
\end{figure}

\begin{remark}
 Recall that transmitting  $r_{u},r_y$ to the
MC-station makes it easier to cope with detection and control of cyber-attacks and faults. However, it implies increasing data transmissions. Furthermore,
transmitting raw data $r_{u}$ increases the risk of eavesdropping attacks
that enable an easy generation of stealthy cyber-attacks. On account of
these concerns, the residual $r_{y,u}$ as a fusion of $r_{y}$ and $%
r_{u},$ as given in (\ref{eq15-38}), is generated and transmitted over the
network. In order to reduce the stress of data transmission, the dimension
of $r_{y,u}$, and thus the demanded channel capacity, is limited by %
$
\dim \left( r_{y,u}\right) =k_y.
$
\end{remark}

\begin{remark}
It is noted that one of the major differences of this work with the existing literatures on the attack detection and resilient control of CPSs  \cite{LWZCY-AUTO-2023,ZHOU2023110723} lies in considering the reference signal $v$. This is motivated by the engineering practice since $v$ is an essential signal that commends the CPS to perform a defined operation. In particular, such a reference signal is event and task depending and will often be online computed as a part of a decision made on the MC-station.
\end{remark}

\subsection{On fault detector design}
On account of the fact that the residual signal $r_y$ is perfectly decoupled from cyber-attacks, an observer-based residual generation system running on the plant side is designed for detecting faults.
In what follows, the fault detection issue is addressed. To this end,
 the following Kalman filter is applied
 \begin{align}
&\hat{x}(k+1)=A\hat{x}(k)+Bu(k)+L(y(k)-\hat{y}(k))\notag\\
&r_{y}(k)=(y(k)-\hat{y}(k)),\hat{y}(k)=C\hat{x}(k)+Du(k)\label{eq-observer-ka}
\end{align}
with
\begin{equation}
\begin{aligned}
 \setlength{\abovedisplayskip}{5pt}\setlength{\belowdisplayskip}{5pt}
&L=APC^T\Sigma_{r_y}^{-1}, \Sigma_{r_y}=CPC^T+\Sigma_{\nu}\\
&P=APA^T+\Sigma_w-APC^T\Sigma_{r_y}^{-1}CPA^T.
\end{aligned}\notag
\end{equation}
As a result, we have $r_y\sim \mathcal{N}(0,\Sigma_{r_y})$. Thus, 
 the following test statistic and threshold with given specification $0<\alpha<1$ can be applied for fault detection purpose
\begin{equation}
 \setlength{\abovedisplayskip}{5pt}\setlength{\belowdisplayskip}{5pt}
J_{rel}=r_y^T \Sigma^{-1}_{r_y} r_y\sim \mathcal{X}^2(k_y),J_{rel,th} =\mathcal{X}_\alpha^2(k_y) \label{eq15-55a}
\end{equation}
where $\mathcal{X}^2(k_y)$ denotes the $\mathcal{X}^2$ distribution with $k_y$ degrees of freedom, and $\mathcal{X}_\alpha^2(k_y)$ is determined based on the chi-squared distribution table.


\begin{remark}
Recall that the estimates $\hat{y}$ contain full information of 
 the nominal system dynamic, and the generated residuals are of the property  $\mathbf{E} (\hat{y}^Tr_y)=0$. In other words, $r_y$ are perpendicular to the estimates and thus do not comprise process information. That means, transmission of residual signals prevents the data privacy efficiently.
\end{remark}

\subsection{On attack detector design}
Distinguishing
from typical technical faults caused by ageing or damage of system components,
cyber-attacks are artificially generated signals injected into the system
over a certain time interval. It is similar to the so-called intermittent
faults \cite{Zhou_review_2020}. In this regard, the detection task comprises not only detecting
injection of cyber-attacks as timely and reliable as possible, but also
detecting switching-off of cyber-attacks. In the sequel, these two tasks are addressed separately. The latter is of a special
importance for the reason that the cyber-attack resilient control law is to
be switched to the nominal controller, once the injected cyber-attacks
disappear. 
%
%

For our purpose, cyber-attack detection is performed at the MC-station under the condition that no fault is detected by the FD system on the plant side. For our purpose, the impact of cyber-attacks on the residual dynamics (\ref{eq-ad-post-filter}) is examined first.

\begin{theorem}\label{th4a}
Given CPS modelled by (\ref{eq14-0a}) with the CPS configuration given in (\ref{eq15-38a})-(\ref{eq15-38c}). With the post-filter as
\begin{equation}
 \setlength{\abovedisplayskip}{5pt}\setlength{\belowdisplayskip}{5pt}
 \bar{R}=R(I-Q_{r,2}Q_{u_{MC}}),R\in \mathcal{RH}_\infty\label{eq-post-filter}
\end{equation}
 the residual dynamics  (\ref{eq-ad-post-filter})  is equivalent to 
\begin{equation}
 \setlength{\abovedisplayskip}{5pt}\setlength{\belowdisplayskip}{5pt}
\bar{r}_{\eta _{a}}=R(Q_{r,1}r_{y}+\eta _{a}),\eta _{a} =Q_{r,2}a_{u_{MC}}+a_{r_{y,u}}.\label{eq16-10a}
\end{equation}
\end{theorem}
\begin{proof}
Recall that 
\begin{equation}
 \setlength{\abovedisplayskip}{5pt}\setlength{\belowdisplayskip}{5pt}
r_{u} =Xu+Yy=u-F\hat{x}=u_{MC}^{a}. \notag
\end{equation}%
As a result, one has that
\begin{align}
&r_u
\hspace{-2pt}=\hspace{-2pt}Q_{u_{MC}}\hspace{-2pt}\left( Q_{r,1}r_{y}\hspace{-2pt}+\hspace{-2pt}Q_{r,2}r_u\hspace{-2pt}+\hspace{-2pt}a_{r_{y,u}}\hspace{-3pt}-\hspace{-2pt}Q_{r,2}v\right) \hspace{-2pt}+\hspace{-2pt}v\hspace{-2pt}+\hspace{-2pt}a_{u_{MC}}\notag
\\
& =\hspace{-2pt}(I\hspace{-2pt}-\hspace{-2pt}Q_{u_{MC}}Q_{r,2})^{-1}\left(Q_{u_{MC}}Q_{r,1}r_{y}+\vartheta _{a}\right)\hspace{-2pt}+\hspace{-2pt}v\label{eq-ruama}
\end{align}
where $\vartheta _{a}=Q_{u_{MC}}a_{r_{y,u}}+a_{u_{MC}}$.
With the aid of the relation
\begin{equation}
 \setlength{\abovedisplayskip}{5pt}\setlength{\belowdisplayskip}{5pt}
Q_{r,2}(I\hspace{-2pt}-\hspace{-2pt}Q_{u_{MC}}Q_{r,2})^{-1}=(I\hspace{-2pt}-\hspace{-2pt}Q_{r,2}Q_{u_{MC}})^{-1}Q_{r,2}
\end{equation}
one has that
\begin{equation}
 \setlength{\abovedisplayskip}{5pt}\setlength{\belowdisplayskip}{5pt}
Q_{r,2}r_{u}
=(I\hspace{-2pt}-\hspace{-2pt}Q_{r,2}Q_{u_{MC}})^{-1}Q_{r,2}\left(Q_{u_{MC}}Q_{r,1}r_{y}\hspace{-2pt}+\hspace{-2pt}\vartheta _{a}\right)\hspace{-2pt}+\hspace{-2pt}Q_{r,2}v.\hspace{-2pt}\notag
\end{equation}
Consequently, we have
\begin{equation}
 \setlength{\abovedisplayskip}{5pt}\setlength{\belowdisplayskip}{5pt}
r_{y,u}^{a}=\left( I\hspace{-2pt}-\hspace{-2pt}Q_{r,2}Q_{u_{MC}}\right) ^{-1}\hspace{-2pt}\left(
Q_{r,1}r_{y}\hspace{-2pt}+\hspace{-2pt}\eta _{a} \right) \hspace{-1pt}+\hspace{-1pt}Q_{r,2}v,  \label{eq16-10}
\end{equation}%
which leads to
\begin{equation}
 \setlength{\abovedisplayskip}{5pt}\setlength{\belowdisplayskip}{5pt}
\bar{r}_{\eta _{a}}=\bar{R}(z)(r_{y,u}^a-Q_{r,2}v)=R(Q_{r,1}r_{y}+\eta _{a})
\end{equation}
and completes the proof.
\end{proof}
It is evident that, $\bar{r}_{\eta _{a}}$ solely contains redundant information about the
plant uncertainties comprised in $r_{y},$ when there exists no cyber-attack.
Consequently, the detection scheme can be designed depending on the
specifications of the plant uncertainties.
Assumed that $Q_{r,1}$ has the following state-space representation 
\begin{equation}
 \setlength{\abovedisplayskip}{2pt}\setlength{\belowdisplayskip}{0pt}
Q_{r,1}=\left( A_{r,1},B_{r,1},C_{r,1},D_{r,1}\right) ,  \label{eq16-13}
\end{equation}%
with $D_{r,1}$ as a full row rank. 
Recalling that any observer (and so Kalman
filter) can be equivalently realized by a post filter, we have the following theorem.
\begin{theorem}\label{th6}
Consider CPS modelled by (\ref{eq14-0a}) with the CPS configuration given in (\ref{eq15-38a})-(\ref{eq15-38c}),  the residual dynamics  (\ref{eq-ad-post-filter}) with the post-filter as (\ref{eq-post-filter}).
Setting
\begin{equation}
 \setlength{\abovedisplayskip}{5pt}\setlength{\belowdisplayskip}{5pt}
R=\left( A_{r,1}-LC_{r,1},-L,\Sigma _{r_{,1}}^{-1/2}C_{r,1},\Sigma
_{r_{,1}}^{-1/2}\right)  \label{eq16-14}
\end{equation}%
with  $L$ as the Kalman filter gain satisfying 
\begin{equation}
 \setlength{\abovedisplayskip}{5pt}\setlength{\belowdisplayskip}{5pt}
\begin{aligned}
&L =\left( A_{r,1}PC_{r,1}^{T}+B_{r,1}\Sigma _{r_{y}}D_{r,1}^{T}\right)
\Sigma _{r_{,1}}^{-1}, \\
&\Sigma _{r_{,1}} =C_{r,1}PC_{r,1}^{T}+D_{r,1}\Sigma _{r_{y}}D_{r,1}^{T},
\end{aligned}%
\end{equation}%
and $P$ as the solution to the following Riccati equation, 
\begin{equation*}
 \setlength{\abovedisplayskip}{5pt}\setlength{\belowdisplayskip}{5pt}
P-A_{r,1}PA_{r,1}^{T}-B_{r,1}\Sigma _{r_{y}}B_{r,1}^{T}+L\Sigma
_{r_{,1}}L^{T}=0
\end{equation*}%
leads to
\begin{align*}
&RQ_{r,1} \hspace{-2pt}=\hspace{-2pt}\left( A_{r,1}\hspace{-2pt}-\hspace{-2pt}LC_{r,1},B_{r,1}\hspace{-2pt}-\hspace{-2pt}LD_{r,1},\Sigma
_{r_{,1}}^{-1}C_{r,1},\Sigma _{r_{,1}}^{-1}D_{r,1}\right) , \\
&\bar{r}_{\eta _{a}} :=RQ_{r,1}r_{y}\sim \mathcal{N}\left( 0,I\right).
\end{align*}%
\end{theorem}

From Theorem \ref{th6} that  the residual  $\bar{r}_{\eta _{a}}$  contains the
necessary information for detecting the cyber-attacks $\left(
a_{u_{MC}},a_{r_{y,u}}\right) .$ Consequently, the $\chi ^{2}$ test statistic 
\begin{equation*}
 \setlength{\abovedisplayskip}{3pt}\setlength{\belowdisplayskip}{3pt}
\chi ^{2}\left( k_y\right) =\bar{r}_{\eta _{a}}^{T}\bar{r}_{\eta _{a}}\overset{%
\mathcal{H}_{0}}{\underset{\mathcal{H}_{1}}{\lessgtr }}J_{th,\chi _{\alpha
}^{2}}
\end{equation*}%
is adopted for the detection end. Thus, Algorithm \ref{al3} can be applied for attack detection.
\begin{algorithm}[!t]
	\caption{An attack detection algorithm}%
	\begin{%
		algorithmic}[1]
	\vspace{0.2cm}
	\State Set $R$ as (\ref{eq16-14}) and $\bar{R}=R(I-Q_{r,2}Q_{u_{MC}})$;
	\State Run the residual generator (\ref{eq-ad-post-filter});
	\State Calculate the test statistic and threshold 
	\begin{equation}
	 \setlength{\abovedisplayskip}{5pt}\setlength{\belowdisplayskip}{5pt}
J=\bar{r}_{\eta _{a}}^T  \bar{r}_{\eta _{a}}\sim \mathcal{X}^2(k_y),J_{th} =\mathcal{X}_\alpha^2(k_y) \label{eq15-55a}
\end{equation}
	\State  Fault detection by means of the detection logic%
\begin{equation}
 \setlength{\abovedisplayskip}{5pt}\setlength{\belowdisplayskip}{0pt}
\text{Detection logic} \text{:}\left\{ 
\begin{array}{l}
J>J_{th}\Longrightarrow \text{cyber-attack} \\ 
J\leq J_{th}\Longrightarrow \text{ attack-free.}%
\end{array}%
\right.  \notag 
\end{equation}
	\end{algorithmic}\label{al3}
\end{algorithm}


\begin{remark}
Both $Q_{r,1},Q_{r,2}$ are, together with $Q_{u_{MC}},$ design parameters
for the controller design. In other words, the conditions determined by
solving the above-defined detection problem are to be considered, when
necessary, as the controller design is addressed. Moreover, although $%
Q_{r,2} $ is not explicitly included in the model (\ref{eq16-10a}), it is a
part of the residual generator  (\ref{eq15-38})  for $r_{y,u}$, on which the
residual $r_{\eta _{a}}$ is built.
\end{remark}

\subsection{On monitoring of cyber-attacks and detecting switching off of
cyber-attacks }
We are now in a position to investigate a scheme to monitor and detect the switching off of cyber-attacks. 
Given the residual model (\ref{eq-post-filter})-(\ref{eq16-10a}), the main objective is to develop a post-filter $R,$ a test statistic $J$ as well as a
threshold so that switching off of cyber-attacks is optimally detected. 
In our study,  it is assumed that residual data $\bar{r}_{\eta _{a}}\left( k_{0}\right) ,\cdots ,\bar{r}_{\eta _{a}}\left(
k_{0}+s\right) $ are available for the detection purpose, where $s$ is an
integer and serves as a hyperparameter.

On account of the resilient control strategy, ``weak attacks" can
be well tolerated. In this context, it is said%
\begin{equation}
 \setlength{\abovedisplayskip}{5pt}\setlength{\belowdisplayskip}{5pt}
\left\{ 
\begin{array}{l}
\left\Vert \eta _{a}\right\Vert _{RMS,\left[ k,k+\tau \right] }\leq L_{\eta
_{a},l}\Longrightarrow \text{attack-free} \\ 
\left\Vert \eta _{a}\right\Vert _{RMS,\left[ k,k+\tau \right] }\geq L_{\eta
_{a},u}\Longrightarrow \text{cyber-attack.}%
\end{array}%
\right.  \label{eq16-17}
\end{equation}%
where $\tau >>s$, and
 $L_{\eta _{a},l},L_{\eta _{a},u}$ are the lower- and upper-bounds, $%
L_{\eta _{a},l}<L_{\eta _{a},u}.$

Now, we begin with the formulation of the detection problem to be addressed.
Consider the model%
\begin{equation*}
 \setlength{\abovedisplayskip}{5pt}\setlength{\belowdisplayskip}{5pt}
\bar{r}_{\eta _{a}}=RQ_{r,1}r_{y}+\bar{\eta}_{a}\sim \mathcal{N}\left( \bar{%
\eta}_{a},I\right) ,\bar{\eta}_{a}=R\eta _{a}\in \mathbb{C}^{m},\eta _{a}\in 
\mathbb{C}^{m},
\end{equation*}%
and write it into 
\begin{align}
\bar{r}_{\eta _{a},s}\left( k_{0}\right) & =\varepsilon _{\eta _{a},s}\left(
k_{0}\right) +\bar{\eta}_{a,s}\left( k_{0}\right) ,\label{eq16-19}\\
\bar{r}_{\eta
_{a},s}\left( k_{0}\right) &=\left[ 
\begin{array}{c}
\bar{r}_{\eta _{a}}\left( k_{0}\right) \\ 
\vdots \\ 
\bar{r}_{\eta _{a}}\left( k_{0}+s\right)%
\end{array}%
\right] ,\bar{\eta}_{a,s}\left(
k_{0}\right) =\left[ 
\begin{array}{c}
\bar{\eta}_{a}\left( k_{0}\right) \\ 
\vdots \\ 
\bar{\eta}_{a}\left( k_{0}+s\right)%
\end{array}%
\right]  \notag \\
\varepsilon _{\eta _{a},s}\left( k_{0}\right) & =\left[ 
\begin{array}{c}
\varepsilon _{\eta _{a}}\left( k_{0}\right) \\ 
\vdots \\ 
\varepsilon _{\eta _{a}}\left( k_{0}+s\right)%
\end{array}%
\right] \sim \mathcal{N}\left( 0,I\right) .  \notag
\end{align}%
By means of the state space realisation (\ref{eq16-14}) of $R$ with $%
x_{R}\left( k\right) $ denoting its state vector, $\bar{\eta}_{a,s}\left(
k_{0}\right) $ is written into%
\begin{align*}
&\bar{\eta}_{a,s}\hspace{-2pt}\left( k_{0}\right) \hspace{-2pt} =\hspace{-2pt}H_{o,s}A_{R,L}^{\gamma }x_{R}\hspace{-2pt}\left(
k_{0}\hspace{-2pt}-\hspace{-2pt}\gamma \right) \hspace{-2pt}+\hspace{-2pt}H_{\bar{\eta}_{a,s,\gamma }}\eta _{a,s+\gamma }\hspace{-2pt}\left(
k_{0}\hspace{-2pt}-\hspace{-2pt}\gamma \hspace{-2pt}+\hspace{-2pt}1\right)\\
&A_{R,L} \hspace{-2pt}=\hspace{-2pt}A_{r,1}\hspace{-2pt}-\hspace{-2pt}LC_{r,1},B_{R}\hspace{-2pt}=\hspace{-2pt}-\hspace{-2pt}L,C_{R}\hspace{-2pt}=\hspace{-2pt}\Sigma
_{r_{,1}}^{-1}C_{r,1},D_{R}\hspace{-2pt}=\hspace{-2pt}\Sigma _{r_{,1}}^{-1}, \\
&H_{o,s} \hspace{-2pt}=\hspace{-2pt}\left[ 
\begin{array}{c}
C_{R} \\ 
C_{R}A_{R,L} \\ 
\vdots \\ 
C_{R}A_{R,L}^{s}%
\end{array}%
\right]\hspace{-2pt},\eta _{a,s+\gamma }\left( k_{0}\hspace{-2pt}-\hspace{-2pt}\gamma \hspace{-2pt}+\hspace{-2pt}1\right) \hspace{-2pt}=\hspace{-4pt}\left[ \hspace{-2pt}
\begin{array}{c}
\eta _{a}\hspace{-2pt}\left( k_{0}\hspace{-2pt}-\hspace{-2pt}\gamma\hspace{-2pt} +\hspace{-2pt}1\right) \\ 
\vdots \\ 
\eta _{a}\hspace{-2pt}\left( k_{0}\hspace{-2pt}+\hspace{-2pt}s\right)%
\end{array}%
\hspace{-2pt}\right]\hspace{-2pt}, \\
&H_{\bar{\eta}_{a,s,\gamma }}\hspace{-4pt}=\hspace{-4pt}\left[ \hspace{-2pt}
\begin{array}{cc}
H_{\bar{\eta}_{a,s,\gamma },1} & H_{\bar{\eta}_{a,s,\gamma },2}%
\end{array}%
\hspace{-2pt}\right]\hspace{-2pt},H_{\bar{\eta}_{a,s,\gamma },1}\hspace{-3pt}=\hspace{-3pt}H_{o,s}\hspace{-3pt}\left[ \hspace{-4pt}
\begin{array}{ccc}
A_{R,L}^{\gamma\hspace{-1pt} -\hspace{-1pt}1} &\hspace{-1pt} \cdots \hspace{-2pt}& \hspace{-1pt}A_{R,L}%
\end{array}%
\hspace{-2pt}\right] \hspace{-3pt}, \\
&H_{\bar{\eta}_{a,s,\gamma },2} \hspace{-2pt}=\hspace{-4pt}\left[ \hspace{-2pt}
\begin{array}{cccc}
D_{R} & 0 & \cdots & 0 \\ 
C_{R}B_{R} & D_{R} & \ddots & \vdots \\ 
\vdots & \ddots & \ddots & 0 \\ 
C_{R}A_{R,L}^{s-1}B_{R} & \cdots & C_{R}B_{R} & D_{R}%
\end{array}%
\hspace{-2pt}\right] \hspace{-4pt}\in \hspace{-2pt}\mathbb{R}^{k_y\left( s+1\right) \times k_y\left( s+1\right) }.
\end{align*}%
On the assumption that the integer $\gamma $ is sufficiently large so that $%
A_{R,L}^{\gamma }\approx 0,$ $\bar{\eta}_{a,s}\left( k_{0}\right) $ is well
approximated by 
\begin{equation*}
\bar{\eta}_{a,s}\left( k_{0}\right) =H_{\bar{\eta}_{a,s,\gamma }}\eta
_{a,s+\gamma }\left( k_{0}-\gamma +1\right) .
\end{equation*}%
Note that $rank\left( H_{\bar{\eta}_{a,s,\gamma }}\right) \leq k_y\left(
s+1\right) .$ In the sequel, without loss of generality, it is assumed that 
\begin{equation*}
 \setlength{\abovedisplayskip}{5pt}\setlength{\belowdisplayskip}{5pt}
rank\left( H_{\bar{\eta}_{a,s,\gamma }}\right) =k_y\left( s+1\right) 
\end{equation*}%
and  
$\eta _{a,s+\gamma }\left( k_{0}-\gamma +1\right) $ doesn't belong to the
null space of $H_{\bar{\eta}_{a,s,\gamma }},$ i.e. 
\begin{equation*}
 \setlength{\abovedisplayskip}{5pt}\setlength{\belowdisplayskip}{5pt}
H_{\bar{\eta}_{a,s,\gamma }}\eta _{a,s+\gamma }\left( k_{0}-\gamma +1\right)
\neq 0.
\end{equation*}%
Moreover, the detection logic (\ref{eq16-17}) is re-formulated as 
\begin{equation}
 \setlength{\abovedisplayskip}{5pt}\setlength{\belowdisplayskip}{5pt}
\left\{ \hspace{-4pt}
\begin{array}{l}
\left\Vert \bar{\eta}_{a,s}\hspace{-2pt}\left( k_{0}\right) \right\Vert ^{2}\hspace{-2pt}\leq \hspace{-2pt} \tau
\bar{\sigma}_{\min}L_{0}^{2}\Longrightarrow \text{attack-free} \\ 
\left\Vert \bar{\eta}_{a,s}\hspace{-2pt}\left( k_{0}\right) \right\Vert ^{2}\hspace{-2pt}>\hspace{-2pt}\tau
\bar{\sigma}_{\max} L_{0}^{2}\hspace{-2pt}\Longrightarrow \text{cyber-attack}%
\end{array}%
\right.  \label{eq16-17a}
\end{equation}%
for $\tau =s+\gamma $ and some constant $L_{0}>0.$ Here, 
\begin{equation*}
 \setlength{\abovedisplayskip}{5pt}\setlength{\belowdisplayskip}{5pt}
\bar{\sigma}_{\max}\hspace{-2pt}=\hspace{-2pt}\lambda _{\max
}\hspace{-2pt}\left(\hspace{-2pt} H_{\bar{\eta}_{a,s,\gamma }}H_{\bar{\eta}_{a,s,\gamma }}^{T}\hspace{-2pt}\right),\bar{\sigma}_{\min}\hspace{-2pt}=\hspace{-2pt}\lambda _{\min }\hspace{-2pt}\left( \hspace{-2pt}H_{\bar{\eta}_{a,s,\gamma }}H_{\bar{\eta}%
_{a,s,\gamma }}^{T}\hspace{-2pt}\right)
\end{equation*}
 are the maximum and minimum eigenvalues of $H_{%
\bar{\eta}_{a,s,\gamma }}H_{\bar{\eta}_{a,s,\gamma }}^{T}$, respectively.

The detection logic (\ref{eq16-17a}) is interpreted as follows. Suppose that 
$\eta _{a,s+\gamma }\left( k_{0}-\gamma +1\right) $ doesn't belong to the
null space of $H_{\bar{\eta}_{a,s,\gamma }},$ i.e. 
$
H_{\bar{\eta}_{a,s,\gamma }}\eta _{a,s+\gamma }\left( k_{0}-\gamma +1\right)
\neq 0.
$
Then, 
\begin{equation*}
 \setlength{\abovedisplayskip}{5pt}\setlength{\belowdisplayskip}{5pt}
\begin{aligned}
&\left\Vert \bar{\eta}_{a,s}\hspace{-3pt}\left( k_{0}\right) \hspace{-1pt}\right\Vert ^{2} \hspace{-3pt}=\hspace{-2pt}\eta
_{a,s+\gamma }^{T}\left( k_{0}\hspace{-3pt}-\hspace{-2pt}\gamma \hspace{-2pt}+\hspace{-2pt}1\right) \hspace{-2pt}H_{\bar{\eta}_{a,s,\gamma
}}^{T}\hspace{-2pt}H_{\bar{\eta}_{a,s,\gamma }}\hspace{-2pt}\eta _{a,s+\gamma }\hspace{-2pt}\left( k_{0}\hspace{-2pt}-\hspace{-2pt}\gamma
\hspace{-2pt}+\hspace{-2pt}1\right) \\
&\left\Vert \bar{\eta}_{a,s}\left( k_{0}\right) \right\Vert ^{2} \geq
\bar{\sigma}_{\min}\eta _{a,s+\gamma }^{T}\left( k_{0}\hspace{-2pt}-\hspace{-2pt}\gamma \hspace{-2pt}+\hspace{-2pt}1\right) \eta
_{a,s+\gamma }\left( k_{0}\hspace{-2pt}-\hspace{-2pt}\gamma \hspace{-2pt}+\hspace{-2pt}1\right) , \\
&\left\Vert \bar{\eta}_{a,s}\left( k_{0}\right) \right\Vert ^{2} \leq
\bar{\sigma}_{\max} \eta _{a,s+\gamma }^{T}\left( k_{0}\hspace{-2pt}-\hspace{-2pt}\gamma \hspace{-2pt}+\hspace{-2pt}1\right) \eta
_{a,s+\gamma }\left( k_{0}\hspace{-2pt}-\hspace{-2pt}\gamma \hspace{-2pt}+\hspace{-2pt}1\right) .
\end{aligned}%
\end{equation*}
Consequently, together with the detection logic (\ref{eq16-17a}),  we have
\begin{equation*}
 \setlength{\abovedisplayskip}{5pt}\setlength{\belowdisplayskip}{5pt}
\begin{aligned}
\left\Vert \eta _{a,s+\gamma }\hspace{-2pt}\left( k_{0}\hspace{-2pt}-\hspace{-2pt}\gamma \hspace{-2pt}+\hspace{-2pt}1\right) \right\Vert
^{2}\hspace{-4pt}\leq \hspace{-2pt}\tau \hspace{-1pt}L_{0}^{2}\Longleftrightarrow \hspace{-4pt}\sqrt{\frac{1}{\tau }\hspace{-2pt}%
\sum\limits_{i=k_{0}\hspace{-1pt}-\hspace{-1pt}\gamma +1}^{k_{0}+s}\hspace{-2pt}\eta _{a}^{T}\hspace{-2pt}\left( i\right) \hspace{-2pt}\eta
_{a}\hspace{-2pt}\left( i\right) }\hspace{-2pt}\leq \hspace{-2pt}L_{0}, \\
\left\Vert \eta _{a,s+\gamma }\hspace{-2pt}\left( k_{0}\hspace{-2pt}-\hspace{-2pt}\gamma \hspace{-2pt}+\hspace{-2pt}1\right) \right\Vert
^{2}\hspace{-2pt}\geq\hspace{-2pt} \tau L_{0}^{2}\hspace{-2pt}\Longleftrightarrow\hspace{-4pt} \sqrt{\frac{1}{\tau }\hspace{-2pt}%
\sum\limits_{i=k_{0}-\gamma +1}^{k_{0}+s}\hspace{-6pt}\eta _{a}^{T}\left( i\right)\hspace{-2pt} \eta
_{a}\left( i\right) }\hspace{-2pt}\geq\hspace{-2pt} L_{0}.
\end{aligned}%
\end{equation*}
Now,  given the model (\ref{eq16-19})
and the log-likelihood ratio (LLR)
\begin{equation*}
 \setlength{\abovedisplayskip}{5pt}\setlength{\belowdisplayskip}{5pt}
\begin{aligned}
& \ln \frac{p_{\left\Vert \bar{\eta}_{a,s}\left( k_{0}\right) \right\Vert
^{2}\leq L_{l}}\left( \bar{r}_{\eta _{a},s}\left( k_{0}\right) \right) }{%
p_{\left\Vert \bar{\eta}_{a,s}\left( k_{0}\right) \right\Vert ^{2}\geq
L_{u}}\left( \bar{r}_{\eta _{a},s}\left( k_{0}\right) \right) }, \\
&L_{l}^{2}:=\tau \bar{\sigma}_{\min} L_{0}^{2},L_{u}^{2}:=\tau \bar{\sigma}_{\max}
L_{0}^{2},
\end{aligned}%
\end{equation*}
the detection problem is formulated as finding the LLR
\begin{equation}
 \setlength{\abovedisplayskip}{5pt}\setlength{\belowdisplayskip}{5pt}
J=\ln \frac{\sup\limits_{\bar{\eta}_{a,s}\left( k_{0}\right) }p_{\left\Vert 
\bar{\eta}_{a,s}\left( k_{0}\right) \right\Vert ^{2}\leq L_{l}^{2}}\left( 
\bar{r}_{\eta _{a},s}\left( k_{0}\right) \right) }{\sup\limits_{\bar{\eta}%
_{a,s}\left( k_{0}\right) }p_{\left\Vert \bar{\eta}_{a,s}\left( k_{0}\right)
\right\Vert ^{2}\geq L_{u}^{2}}\left( \bar{r}_{\eta _{a},s}\left(
k_{0}\right) \right) }  \label{eq16-16}
\end{equation}%
as the test statistic and the corresponding threshold so that 
$
J\overset{\mathcal{H}_{0}}{\underset{\mathcal{H}_{1}}{\lessgtr }}J_{th},
$ 
where the null hypothesis $\mathcal{H}_{0}$ is for attacks, i.e. $\left\Vert
\eta _{a}\right\Vert _{RMS,\left[ k_{0}-\gamma +1,k_{0}-\gamma +\tau \right]
}\geq L_{0},$ and the alternative hypothesis is for attack-free, $\left\Vert
\eta _{a}\right\Vert _{RMS,\left[ k+1,k+\tau \right] }<L_{0}.$ Below, the
solution is outlined.

\begin{theorem}
Considering the model (\ref{eq16-19}), the solution to (\ref{eq16-16}) is given by 
\begin{equation}
 \setlength{\abovedisplayskip}{2pt}\setlength{\belowdisplayskip}{0pt}
J\hspace{-2pt}=\hspace{-2pt}\left\{ \hspace{-2pt}
\begin{array}{l}
\frac{1}{2}\hspace{-2pt}\left( \left\Vert \bar{r}_{\eta _{a},s}\left( k_{0}\right)
\right\Vert \hspace{-2pt}-\hspace{-2pt}L_{u}\right) ^{2},\left\Vert \bar{\eta}_{a,s}\left(
k_{0}\right) \right\Vert ^{2}\leq L_{l}^{2} \\ 
\frac{1}{2}\hspace{-2pt}\left( \hspace{-2pt}
\begin{array}{c}
\left( \left\Vert \bar{r}_{\eta _{a},s}\left( k_{0}\right) \right\Vert
-L_{u}\right) ^{2} \\ 
-\left( \left\Vert \bar{r}_{\eta _{a},s}\left( k_{0}\right) \right\Vert
-L_{l}\right) ^{2}%
\end{array}%
\hspace{-2pt}\right)\hspace{-2pt} ,L_{l}^{2}\hspace{-2pt}\leq \hspace{-2pt}\left\Vert \bar{\eta}_{a,s}\left( k_{0}\right)
\right\Vert ^{2}\hspace{-2pt}\leq \hspace{-2pt}L_{u}^{2} \\ 
-\frac{1}{2}\left( \left\Vert \bar{r}_{\eta _{a},s}\left( k_{0}\right)
\right\Vert -L_{l}\right) ^{2},\left\Vert \bar{\eta}_{a,s}\left(
k_{0}\right) \right\Vert ^{2}\geq L_{u}^{2}.%
\end{array}%
\right.  \notag
\end{equation}%
\end{theorem}
\begin{proof}
 To begin with, $%
J$ in  (\ref{eq16-16}) is written into 
\begin{equation*}
 \setlength{\abovedisplayskip}{5pt}\setlength{\belowdisplayskip}{5pt}
\begin{aligned}
J=\frac{1}{2}\sup\limits_{\substack{ \bar{\eta}_{a,s}\left( k_{0}\right)  \\ %
\left\Vert \bar{\eta}_{a,s}\left( k_{0}\right) \right\Vert ^{2}\leq
L_{l}^{2} }}-\left\Vert \bar{r}_{\eta _{a},s}\left( k_{0}\right) -\bar{\eta}%
_{a,s}\left( k_{0}\right) \right\Vert ^{2} \\
-\frac{1}{2}\sup\limits_{\substack{ \bar{\eta}_{a,s}\left( k_{0}\right)  \\ %
\left\Vert \bar{\eta}_{a,s}\left( k_{0}\right) \right\Vert ^{2}\geq
L_{u}^{2} }}-\left\Vert \bar{r}_{\eta _{a},s}\left( k_{0}\right) -\bar{\eta}%
_{a,s}\left( k_{0}\right) \right\Vert ^{2},
\end{aligned}%
\end{equation*}
which can be solved equivalently by means of the following two optimisation
problems, 
\begin{equation*}
 \setlength{\abovedisplayskip}{5pt}\setlength{\belowdisplayskip}{5pt}
 \frac{1}{2}\inf\limits_{_{\substack{ \bar{\eta}_{a,s}\left( k_{0}\right) 
\\ \left\Vert \bar{\eta}_{a,s}\left( k_{0}\right) \right\Vert ^{2}\leq
L_{l}^{2}}}}\left\Vert \bar{r}_{\eta _{a},s}\left( k_{0}\right) -\bar{\eta}%
_{a,s}\left( k_{0}\right) \right\Vert ^{2}, 
\end{equation*}
\begin{equation*}
 \frac{1}{2}\inf\limits_{_{\substack{ \bar{\eta}_{a,s}\left( k_{0}\right) 
\\ \left\Vert \bar{\eta}_{a,s}\left( k_{0}\right) \right\Vert ^{2}\geq
L_{u}^{2}}}}\left\Vert \bar{r}_{\eta _{a},s}\left( k_{0}\right) -\bar{\eta}%
_{a,s}\left( k_{0}\right) \right\Vert ^{2}.
\end{equation*}
Accordingly, it turns out, for $\left\Vert \bar{\eta}_{a,s}\left(
k_{0}\right) \right\Vert ^{2}\leq L_{l}^{2},$ 
\begin{equation*}
 \setlength{\abovedisplayskip}{5pt}\setlength{\belowdisplayskip}{5pt}
J\hspace{-2pt}=\hspace{-2pt}\frac{1}{2}\hspace{-2pt}\left\Vert \bar{r}_{\eta _{a},s}\hspace{-2pt}\left( k_{0}\right) \hspace{-2pt}-\hspace{-2pt}\hat{\eta}%
_{a,s}^{u}\hspace{-2pt}\left( k_{0}\right) \hspace{-2pt}\right\Vert ^{2}\hspace{-2pt}, \hat{\eta}_{a,s}^{u}\hspace{-2pt}\left( k_{0}\right) \hspace{-3pt}=\hspace{-3pt}\frac{\bar{r}_{\eta _{a},s}\left(
k_{0}\right) L_{u}}{\sqrt{\bar{r}_{\eta _{a},s}^{T}\hspace{-2pt}\left( k_{0}\right) \hspace{-2pt}\bar{r%
}_{\eta _{a},s}\hspace{-2pt}\left( k_{0}\right) }},
\end{equation*}%
for $\left\Vert \bar{\eta}_{a,s}\left( k_{0}\right) \right\Vert ^{2}\geq
L_{u}^{2},$%
\begin{equation*}
 \setlength{\abovedisplayskip}{5pt}\setlength{\belowdisplayskip}{5pt}
J\hspace{-2pt}=\hspace{-2pt}-\hspace{-2pt}\frac{1}{2}\hspace{-2pt}\left\Vert \hspace{-1pt}\bar{r}_{\eta _{a},s}\hspace{-2pt}\left( k_{0}\right) \hspace{-2pt}-\hspace{-2pt}\hat{\eta}%
_{a,s}^{l}\hspace{-2pt}\left( k_{0}\right) \hspace{-2pt}\right\Vert ^{2}\hspace{-2pt}, \hat{\eta}_{a,s}^{l}\hspace{-2pt}\left( k_{0}\right)\hspace{-2pt} =\hspace{-4pt}\frac{\bar{r}_{\eta _{a},s}\left(
k_{0}\right) L_{l}}{\sqrt{\bar{r}_{\eta _{a},s}^{T}\hspace{-2pt}\left( k_{0}\right)\hspace{-2pt} \bar{r%
}_{\eta _{a},s}\hspace{-2pt}\left( k_{0}\right) }},
\end{equation*}%
and for $L_{l}^{2}\leq \left\Vert \bar{\eta}_{a,s}\left( k_{0}\right)
\right\Vert ^{2}\leq L_{u}^{2},$%
\begin{equation*}
J\hspace{-2pt}=\hspace{-2pt}\frac{1}{2}\hspace{-2pt}\left\Vert \bar{r}_{\eta _{a},s}\left( k_{0}\right) \hspace{-2pt}-\hspace{-2pt}\hat{\eta}%
_{a,s}^{u}\left( k_{0}\right) \right\Vert ^{2}\hspace{-2pt}-\hspace{-2pt}\frac{1}{2}\hspace{-2pt}\left\Vert \bar{r}%
_{\eta _{a},s}\left( k_{0}\right) \hspace{-2pt}-\hspace{-1pt}\hat{\eta}_{a,s}^{l}\left( k_{0}\right)
\right\Vert ^{2}\hspace{-2pt}.
\end{equation*}%
After some routine calculations,  the proof is completed.
\end{proof}

Next, a threshold $J_{th}$ is determined under the condition of a (maximum)
false alarm rate (FAR) equal to $\alpha.$ That means 
\begin{equation*}
\sup\limits_{\left\Vert \bar{\eta}_{a,s}\left( k_{0}\right) \right\Vert
^{2}\geq L_{u}^{2}}\hspace{-3pt}\Pr\hspace{-2pt} \left( J\hspace{-2pt}>\hspace{-2pt}J_{th}\hspace{-2pt}\left\vert \bar{r}_{\eta _{a},s}\left(
k_{0}\right)\hspace{-2pt} \sim \hspace{-2pt}\mathcal{N}\left( \bar{\eta}_{a,s}\left( k_{0}\right)
,I\right) \right. \hspace{-2pt}\right)\hspace{-2pt} =\hspace{-2pt}\alpha .
\end{equation*}%
Observe that the core of the LLR $J$ is $\left( \left\Vert 
\bar{r}_{\eta _{a},s}\left( k_{0}\right) \right\Vert -L_{\mu }\right) ^{2},$ 
$\mu =l,u.$ Let 
\begin{equation*}
\left( \left\Vert \bar{r}_{\eta _{a},s}\left( k_{0}\right) \right\Vert
-L_{\mu }\right) ^{2}=J_{\mu }\left( \bar{r}_{\eta _{a},s}\left(
k_{0}\right) \right) .
\end{equation*}%
Then, 
\begin{equation*}
\left( J_{\mu }^{1/2}\left( \bar{r}_{\eta _{a},s}\left( k_{0}\right) \right)
+L_{\mu }\right) ^{2}=\left\Vert \bar{r}_{\eta _{a},s}\left( k_{0}\right)
\right\Vert ^{2}.
\end{equation*}%
Recall that $\left\Vert \bar{r}_{\eta _{a},s}\left( k_{0}\right) \right\Vert
^{2}$ is subject to a non-central $\chi ^{2}$ distribution with $\left(
s+1\right) m$ degrees of freedom and the non-centrality parameter $%
\left\Vert \bar{\eta}_{a,s}\left( k_{0}\right) \right\Vert ^{2}.$ Since $%
\bar{\eta}_{a,s}\left( k_{0}\right) $ is unknown and no a prior information
of $\bar{\eta}_{a,s}\left( k_{0}\right) $ is available, it is suggested to
calculate the FAR with respect to 
$
\left\Vert \bar{\eta}_{a,s}\left( k_{0}\right) \right\Vert ^{2}=L_{u}^{2}.
$
This means, we check the FAR corresponding to the lower-bound of the $%
\left\Vert \bar{\eta}_{a,s}\left( k_{0}\right) \right\Vert ^{2},$ when
attacks exist. In a certain sense, this is the worst case. As a result, we
finally have%
\begin{equation*}
\left( J_{\mu }^{1/2}\hspace{-2pt}\left( \bar{r}_{\eta _{a},s}\left( k_{0}\right) \right)
\hspace{-3pt}+\hspace{-2pt}L_{\mu }\right) ^{2}\hspace{-3pt}=\hspace{-2pt}\left\Vert \bar{r}_{\eta _{a},s}\hspace{-2pt}\left( k_{0}\right)
\right\Vert ^{2}\hspace{-2pt}\sim \hspace{-2pt}\chi ^{2}\hspace{-2pt}\left( \left( s\hspace{-2pt}+\hspace{-2pt}1\right) m,L_{u}^{2}\right) ,
\end{equation*}%
and based on it, the threshold $J_{th}$ can be calculated.

\subsection{On the resilient and fault-tolerant control}
To this end, we first re-examined the process dynamics for the CPS configuration (\ref{eq15-38a})-(\ref{eq15-38c}).  
\begin{corollary}
Given CPS modelled by (\ref{eq14-0a}) with the CPS configuration given in (\ref{eq15-38a})-(\ref{eq15-38c}), then the closed-dynamics is governed by 
\begin{align}
&\hspace{-6pt}\left[ \hspace{-2pt}
\begin{array}{c}
u \\ 
y%
\end{array}%
\hspace{-2pt}\right] 
\hspace{-4pt}=\hspace{-4pt}\left[ \hspace{-2pt}
\begin{array}{c}
M \\ 
N%
\end{array}%
\hspace{-2pt}\right] \hspace{-4pt}\left(\hspace{-2pt} v\hspace{-2pt}+\hspace{-2pt}\left( I\hspace{-3pt}-\hspace{-3pt}Q_{u_{MC}}Q_{r,2}\right) ^{-\hspace{-1pt}1}\hspace{-3pt}\vartheta
_{a}\hspace{-2pt}\right) \hspace{-4pt}+\hspace{-4pt}\left( \hspace{-2pt}\left[\hspace{-3pt} 
\begin{array}{c}
-\hspace{-2pt}\hat{Y} \\ 
\hat{X}%
\end{array}%
\hspace{-3pt}\right] \hspace{-4pt}+\hspace{-4pt}\left[ \hspace{-2pt}
\begin{array}{c}
M \\ 
N%
\end{array}%
\hspace{-2pt}\right] \hspace{-3pt}Q_{u_{MC}}\hspace{-1pt}\bar{Q}_{r,1}\hspace{-4pt}\right) \hspace{-3pt}r_{y}\notag \\
&\vartheta _{a}\hspace{-2pt}:\hspace{-1pt}=\hspace{-2pt}Q_{u_{MC}}a_{r_{y,u}}\hspace{-2pt}+\hspace{-2pt}a_{u_{MC}},\bar{Q}_{r,1}\hspace{-2pt}=\hspace{-2pt}(I\hspace{-2pt}-\hspace{-2pt}{Q}_{r,2}Q_{u_{MC}})^{-1}\hspace{-1pt}{Q}_{r,1} \label{eq16-20}
\end{align}%
\end{corollary}
\begin{proof}
It follows from the proof of Theorem \ref{th4a} that
\begin{align}
&u_{MC}^{a}
\hspace{-2pt} =\hspace{-2pt}v\hspace{-2pt}+\hspace{-2pt}(I\hspace{-2pt}-\hspace{-2pt}Q_{u_{MC}}Q_{r,2})^{-1}\left(Q_{u_{MC}}Q_{r,1}r_{y}+\vartheta _{a}\right).\label{eq-ruam}
\end{align}
Since 
\begin{equation}
 \setlength{\abovedisplayskip}{2pt}\setlength{\belowdisplayskip}{1pt}
\left[ 
\begin{array}{c}
u \\ 
y%
\end{array}%
\right] =\left[ 
\begin{array}{cc}
X & \text{ }Y \\ 
-\hat{N} & \text{ }\hat{M}%
\end{array}%
\right] ^{-1}\left[ 
\begin{array}{c}
u_{MC}^a \\ 
r_{y}%
\end{array}%
\right] 
\end{equation}
 the closed-loop dynamic is described by (\ref{eq16-20}).
\end{proof}

Recall that  the faults and the
corresponding control performance degradations in the plant solely cause
variations in the system residual subspace and the attacks only lead to changes in the system image subspace. 
Based on it, we are now in the position to propose a  resilient and FTC scheme.  Given $1>\gamma
_{r_{y}}>0$ and $1>\gamma _{\vartheta _{a}}>0$ that represent the resilient and FTC performance, find $Q_{r,1},Q_{r,2}$ and $Q_{u_{MC}}$ such that $%
\left( I-Q_{u_{MC}}Q_{r,2}\right) ^{-1}$ is stable, and%
\begin{align}
&\left\Vert \left( I-Q_{u_{MC}}Q_{r,2}\right) ^{-1}\right\Vert _{\infty }
\leq \gamma _{\vartheta _{a}},  \label{eq16-21a} \\
&\left\Vert \left[ 
\begin{array}{c}
-\hat{Y} \\ 
\hat{X}%
\end{array}%
\right] +\left[ 
\begin{array}{c}
M \\ 
N%
\end{array}%
\right] Q_{u_{MC}}\bar{Q}_{r,1}\right\Vert _{\infty } \leq \gamma _{r_{y}}.
\label{eq16-21b}
\end{align}%

\begin{remark}
It is plaint that the condition (\ref{eq16-21b}) describes the system
robustness against the potential faults presented by the residual $r_{y},$
while the condition (\ref{eq16-21a}) imposes the resilient requirement
dealing with the cyber-attack $\vartheta _{a}.$ Observe that for the
resilient controller subject to (\ref{eq16-21a})-(\ref{eq16-21b}), $%
Q_{r,1},Q_{r,2}$ and $Q_{u_{MC}}$ are redundant. In fact, $Q_{r,2}$ can be
used as an encoder to complicate the construction of a stealthy attack
vector $\eta _{a}$ defined in (\ref{eq16-10a}). In this case, $Q_{u_{MC}}$
is firstly determined given $Q_{r,2},$ and then $Q_{r,1}$ is found for given 
$Q_{u_{MC}}.$
\end{remark}

%

\section{Stealthy attacks and performance-based detection}\label{sec5}
Stealthy attacks are such attacks that will not cause detectable changes in the detection test
statistic aiming at degrading the system control performance. 
From the defender's point of view, it is of importance to study the potential stealthy attacks, and more importantly, develop the corresponding detection schemes. 
To this end, the stealthy attack design against tracking behaviour and feedback control performance are studied in this section. It is followed by the associated detection schemes.

\subsection{Stealthy attack design against tracking behaviour}
The definition of stealthy attacks is given first.
\begin{definition}
An attack is said to be stealthy if the residual signals satisfy
\begin{equation}
 \setlength{\abovedisplayskip}{5pt}\setlength{\belowdisplayskip}{0pt}
r_{y,u}^a\hspace{-2pt}\sim \hspace{-2pt}\mathcal{N}\hspace{-2pt}\left(\hspace{-2pt} \mathbb{E}r_{y,u}^{a},\Sigma_{r_{y,u}^a}\hspace{-2pt}\right)\hspace{-2pt},\Sigma_{r_{y,u}^a}\hspace{-3pt}=\hspace{-2pt}\Sigma_{r_{y,u}},\mathbb{E}r_{y,u}^{a}\hspace{-3pt}=\hspace{-2pt}\mathbb{E}r_{y,u}
\label{eq-stealthy-condition}\end{equation}
\end{definition}
It
follows from the closed-loop dynamic (\ref{eq16-20}), that the cyber-attack $\vartheta _{a}$
degrades the system tracking behaviour, while a change of the feedback
mechanism of $r_{y}$ will lead to feedback performance degradation. 
Examining the dynamics of the residual generator (\ref{eq16-10a}) leads to the following lemma. 
\begin{lemma}
Consider the CPS configuration (\ref{eq15-38a})-(\ref{eq15-38c}). The following attack is stealthy
\begin{equation}
 \setlength{\abovedisplayskip}{5pt}\setlength{\belowdisplayskip}{5pt}
\eta _{a}=Q_{r,2}a_{u_{MC}}+a_{r_{y,u}}=0. \label{eq-stealthy}
\end{equation}%
\end{lemma}


It is apparent that, to attack the system tracking behaviour stealthily, the
attack design is to be formulated as maximising the tracking error, 
\begin{equation*}
 \setlength{\abovedisplayskip}{5pt}\setlength{\belowdisplayskip}{5pt}
\left[ 
\begin{array}{c}
\Delta u \\ 
\Delta y%
\end{array}%
\right] \hspace{-2pt}:=\hspace{-2pt}\left[ 
\begin{array}{c}
u \\ 
y%
\end{array}%
\right] \hspace{-2pt}-\hspace{-2pt}\left[ 
\begin{array}{c}
M \\ 
N%
\end{array}%
\right] v\hspace{-2pt}=\hspace{-2pt}\left[ 
\begin{array}{c}
M \\ 
N%
\end{array}%
\right] \left( I\hspace{-2pt}-\hspace{-2pt}Q_{u_{MC}}Q_{r,2}\right) ^{-1}\vartheta _{a},
\end{equation*}%
by selecting 
$
\vartheta _{a}=Q_{u_{MC}}a_{r_{y,u}}+a_{u_{MC}}
$
 subject to 
(\ref{eq-stealthy}). 
From the condition (\ref{eq-stealthy}), the following relations
automatically arise, 
\begin{equation}
 \setlength{\abovedisplayskip}{5pt}\setlength{\belowdisplayskip}{5pt}
\vartheta _{a}\hspace{-2pt}=\hspace{-2pt}\left( I\hspace{-2pt}-\hspace{-2pt}Q_{u_{MC}}Q_{r,2}\right) a_{u_{MC}}\hspace{-2pt}\Longrightarrow \hspace{-2pt}
\left[ 
\begin{array}{c}
\Delta u \\ 
\Delta y%
\end{array}%
\right] \hspace{-2pt}=\hspace{-2pt}\left[ 
\begin{array}{c}
M \\ 
N%
\end{array}%
\right] a_{u_{MC}} \label{eq16-24}
\end{equation}%
which indicates that the tracking performance degradation solely depends on $a_{u_{MC}}$.


\begin{remark}
It is interesting to notice that
the  condition (\ref{eq-stealthy}) is similar to the so-called covert
attacks \cite{Smith2015}. The difference between them lies in the transfer
function $Q_{r,2}$, which is the transfer function of the plant in case of a
covert attack. Specifically, $Q_{r,2}$ in (\ref{eq-stealthy}) is a design parameter and can also be
online updated during system operations, while the transfer function of the
plant is fixed. This fact makes the realization of such stealthy attacks
more difficult.
\end{remark}

\subsection{Stealthy attack design against feedback control performance}
Remember that cyber-attacks cause no change in the residual $r_{y}.$ This
implies that degrading the system feedback dynamics can only be achieved by
changing the following feedback mechanism 
\begin{equation}
 \setlength{\abovedisplayskip}{5pt}\setlength{\belowdisplayskip}{5pt}
\left[ 
\begin{array}{c}
-\hat{Y} \\ 
\hat{X}%
\end{array}%
\right] +\left[ 
\begin{array}{c}
M \\ 
N%
\end{array}%
\right] Q_{u_{MC}}\bar{Q}_{r,1}.\label{eq-fb-st}
\end{equation}%
Thus, the stealthy attack design issue against feedback control performance can be formulated as follows:

\emph{Given $r_{y,u}\sim \mathcal{N}\left( 0,\Sigma _{r_{y,u}}\right), r_{y,u}\in \mathcal{R}^{k_y}$ denoting the nominal (attack-free) residual, find a linear mapping of $r_{y,u}$} 
\begin{equation}
 \setlength{\abovedisplayskip}{5pt}\setlength{\belowdisplayskip}{5pt}
r_{y,u}^a=f(r_{y,u})\in \mathcal{R}^{k_y}
\end{equation}
\emph{such that (\ref{eq-stealthy-condition}) holds and
the feedback performance (\ref{eq-fb-st}) is deteriorated  as much as possible, that is, $Q_{u_{MC}}\bar{Q}_{r,1}$ is modified so that the norm of  (\ref{eq-fb-st}) becomes larger.}

For our purpose, the following attack model is considered
\begin{align}
r_{y,u}^{a}=\Pi_a \left( r_{y,u}-\zeta \right) +\zeta \label{eq-sa-de}
\end{align}
where $\Pi _{a}$ is the attack matrix to be designed.
Below,  a  realistic and practical scenario for
constructing stealthy cyber-attacks is developed that degrade the feedback control
performance.
\begin{assumption}
It is assumed that the CPS system under consideration is operating
in the steady state with a constant reference vector in the attack-free case so that the residual vector $r_{y,u}$
is a weakly stationary stochastic process as
\begin{align*}
&\mathbb{E}r_{y,u} =\lim_{z\rightarrow 1}Q_{r,2}v=:\zeta , \Sigma _{r_{y,u}} =\mathbb{E}\left( r_{y,u}-\zeta \right) \left(
r_{y,u}-\zeta \right) ^{T}, \\
&r_{y,u}-\zeta  =\left( I-Q_{r,2}Q_{u_{MC}}\right) ^{-1}Q_{r,1}r_{y},
\end{align*}%
where $\zeta $ and $\Sigma _{r_{y,u}}$ are  constant vector and matrix,
respectively. 
\end{assumption}
Suppose that attackers collect (sufficient) $N$ number of data
of $r_{y,u}$ over the time interval $\left[ k_{1},k_{N}\right] $ and
calculate the maximal likelihood estimates of $\zeta $ and $\Sigma
_{r_{y,u}},$%
\begin{equation*}
 \setlength{\abovedisplayskip}{5pt}\setlength{\belowdisplayskip}{5pt}
\hat{\zeta}\hspace{-2pt}=\hspace{-2pt}\frac{1}{N}\hspace{-2pt}\dsum\limits_{i=1}^{N}\zeta \left( k_{i}\right) \hspace{-2pt},\hat{%
\Sigma}_{r_{y,u}}\hspace{-4pt}=\hspace{-4pt}\frac{1}{N}\hspace{-2pt}\dsum\limits_{i=1}^{N}\hspace{-2pt}\left( r_{y,u}\left(
k_{i}\right) \hspace{-2pt}-\hspace{-2pt}\hat{\zeta}\right) \left( r_{y,u}\left( k_{i}\right) \hspace{-2pt}-\hspace{-2pt}\hat{%
\zeta}\right) ^{T}\hspace{-2pt}.
\end{equation*}%

\begin{theorem}
Consider the CPS configuration (\ref{eq15-38a})-(\ref{eq15-38c}), 
the stealthy attacks can be constructed as follows:%
\begin{equation}
 \setlength{\abovedisplayskip}{5pt}\setlength{\belowdisplayskip}{5pt}
r_{y,u}^{a}= \Pi_a \left( r_{y,u}-\hat{\zeta} \right) +\hat{\zeta}, \Pi_a=\Xi \Pi \label{eq-sa-de1}
\end{equation}%
where $\Pi $ is a Kalman filter satisfying%
\begin{gather*}
\hat{x}_{\Pi }(k+1\left\vert k\right. )=A_{\Pi }\hat{x}_{\Pi }(k\left\vert
k-1\right. )+K(k)\Delta r_{y,u}(k), \\
\Delta r_{y,u}(k)=r_{y,u}(k)-\hat{\zeta} -C_{\Pi }\hat{x}_{\Pi }(k\left\vert
k-1\right. ),\hat{x}_{\Pi }(0)=0, \\
K(k)=A_{\Pi }P(k\left\vert k-1\right. )C_{\Pi }^{T}\Sigma _{\Delta
r_{y,u}}^{-1}(k), \\
P(k+1\left\vert k\right. )=A_{\Pi }P(k\left\vert k-1\right. )A_{\Pi
}^{T}-K(k)\Sigma _{\Delta r_{y,u}}(k)K^{T}(k), \\
\Sigma _{\Delta r_{y,u}}(k)=C_{\Pi }P(k\left\vert k-1\right. )C_{\Pi }^{T}+%
\hat{\Sigma}_{r_{y,u}},
\end{gather*}%
and $\Xi $ is a matrix 
\begin{equation*}
 \setlength{\abovedisplayskip}{5pt}\setlength{\belowdisplayskip}{5pt}
\Xi =U\hat{\Sigma}_{r_{y,u}}^{1/2}\Sigma _{\Delta
r_{y,u}}^{-1/2},UU^{T}=I.
\end{equation*}%
\end{theorem}
\begin{proof}
It is easy to examine that 
\begin{align*}
\mathbb{E}r_{y,u}^{a}& =\mathbb{E}r_{y,u}=\zeta , \Sigma _{r_{y,u}^{a}} =\mathbb{E}\left( r_{y,u}^{a}-\zeta \right) \left(
r_{y,u}^{a}-\zeta \right) ^{T}=\Sigma _{r_{y,u}}
\end{align*}%
so that the added cyber-attack is (strictly) stealthy. \end{proof}
\begin{remark}
Notice that the attack (\ref{eq-sa-de1}) is stealthy
independent of the test statistic implemented on the attack detection, and
hence holds also for Kullback-Leibler divergence (KLD) based test statistic 
\cite{Yang2022attacks,ZhangTAC2023}.
\end{remark}
 Moreover, the
feedback control performance is govered by
\begin{align}
&\left[ 
\begin{array}{c}
u_f \\ 
y_f%
\end{array}%
\right]=\left( \left[ 
\begin{array}{c}
-\hat{Y} \\ 
\hat{X}%
\end{array}%
\right] +\left[ 
\begin{array}{c}
M \\ 
N%
\end{array}%
\right] Q_{u_{MC}}\Pi_a \bar{Q}_{r,1}^a\right) r_{y}\label{eq-sa-fd}\\
&\bar{Q}_{r,1}^a=(I-Q_{r,2}Q_{u_{MC}}\Pi_a)^{-1}Q_{r,1}\notag
\end{align}%
which implies that the cyber-attack (\ref{eq-sa-de1}) can degrade the system feedback control performance. 
Consequently, setting $\Pi_a=\Xi \Pi =-I$ could considerably degrade the feedback
performance. In other words, without any process model knowledge,
constructing the stealthy cyber-attack simply as%
\begin{align*}
r_{y,u}^{a}& =-\left( r_{y,u}-\zeta \right) +\zeta =-r_{y,u}+2\zeta
\Longrightarrow  \\
\left[ 
\begin{array}{c}
u_f \\ 
y_f%
\end{array}%
\right] & \hspace{-2pt}=\hspace{-2pt}\left( \left[ 
\begin{array}{c}
-\hspace{-2pt}\hat{Y} \\ 
\hat{X}%
\end{array}%
\right] \hspace{-2pt}-\hspace{-2pt}\left[ 
\begin{array}{c}
M \\ 
N%
\end{array}%
\right] \hspace{-2pt}Q_{u_{MC}}(I\hspace{-2pt}+\hspace{-2pt}Q_{r,2}Q_{u_{MC}})^{-1}Q_{r,1}\right) r_{y}
\end{align*}%
can result in a remarkable degradation of the system robustness.

\begin{remark}
The cyber-attacks $\Pi_a r_y$ has been well studied with $\Pi_a$ as a unitary matrix
in \cite%
{LIU2022110079,ZhangTAC2023,REN2023110895,ZHOU2023110723,Jin2024}.  It is noteworthy that availability of process model knowledge or collection
of sufficient amount of data, by which process model can be identified, are
the widely accepted and adopted assumptions, on which the developed design
methods for stealthy attacks are based.
On the other hand, attacks on the estimation performance of the embedded observer
are the focus of the existing research, both for control and (remote)
monitoring of CPSs \cite{DeruiDing2021-survey,Ferreira-survey-2024}, while limit attention has been made on attack against feedback control performance.
\end{remark}

\subsection{Performance degradation monitoring-based detection of stealthy attacks against feedback performance}
In the subsequent subsections, we are devoted to develop  detection schemes for monitoring the performance degradation caused by the stealthy attacks against feedback performance and tracking behavior, respectively.

Without loss of generality,  the stealthy attacks on the system feedback result in the change on feedback dynamics as (\ref{eq-sa-fd}). 
That means, the signal $r_{PD}$ generated by 
\begin{equation}
 \setlength{\abovedisplayskip}{5pt}\setlength{\belowdisplayskip}{5pt}
r_{PD}=\Psi \left( \left[ 
\begin{array}{c}
-\hat{Y} \\ 
\hat{X}%
\end{array}%
\right] +\left[ 
\begin{array}{c}
M \\ 
N%
\end{array}%
\right] Q_{u_{MC}}\Pi _{a}\bar{Q}_{r,1}^a\right) r_{y}  \label{eq16-29}
\end{equation}%
is the information carrier of degradation in feedback control performance
that is capable for detecting performance degradation caused by
cyber-attacks. Here, $\Psi \in \mathcal{RH}_{\infty }$ is an arbitrary
stable dynamic system.

Motivated by the aforementioned discussion and the residual (\ref{eq16-29}),
we now propose 
to construct $r_{y,u}$ as follows,%
\begin{equation}
 \setlength{\abovedisplayskip}{3pt}\setlength{\belowdisplayskip}{3pt}
r_{y,u}=Q_{r,1}r_{y}+Q_{r,2}r_{u}+\bar{r}_{PD}, \bar{r}_{PD}=\Psi \left[ 
\begin{array}{c}
u \\ 
y%
\end{array}%
\right] \label{eq16-32}
\end{equation}%

Corresponding to the nominal control law (\ref{eq15-38}%
), the controller is now set as 
\begin{align}
&u =F\hat{x}+u_{MC}-Q_{u_{MC}}\bar{r}_{PD},\notag \\
&u_{MC} =Q_{u_{MC}}(r_{y,u}^{a}-Q_{r,2}v)+v. \label{eq16-32a}
\end{align}%
Let 
\begin{equation}
 \setlength{\abovedisplayskip}{5pt}\setlength{\belowdisplayskip}{5pt}
r_{y,u}^{a} =\Pi _{a}\left( r_{y,u}-Q_{v}v\right) +Q_{v}v, Q_{v} :=\Psi \left[ 
\begin{array}{c}
M \\ 
N%
\end{array}%
\right] v\label{eq16-32b}
\end{equation}%
with $\Pi _{a}$ as the cyber-attack system, and assume that attackers are in
possession of knowledge of $Q_{v}v$ or $Q_{v}v$ is sufficiently small so
that it can be neglected. 

The
subsequent work is dedicated to the analysis of residual dynamics under such
cyber-attacks and, based on it, the development of a performance degradation
detection scheme.
It follows from  (\ref{eq16-32})-(\ref{eq16-32b}) that
\begin{align*}
\left[ 
\begin{array}{c}
u \\ 
y%
\end{array}%
\right] \hspace{-2pt}=\hspace{-2pt}\left[ \hspace{-2pt}
\begin{array}{c}
M \\ 
N%
\end{array}%
\hspace{-2pt}\right] \hspace{-2pt}(I\hspace{-2pt}+\hspace{-2pt}Q_{u_{MC}}\bar{\Pi}_a)v\hspace{-2pt}+\hspace{-2pt}\left( \hspace{-2pt}\left[ \hspace{-2pt}
\begin{array}{c}
-\hspace{-2pt}\hat{Y} \\ 
\hat{X}%
\end{array}%
\hspace{-2pt}\right] \hspace{-2pt}+\hspace{-2pt}\left[ \hspace{-2pt}
\begin{array}{c}
M \\ 
N%
\end{array}%
\hspace{-2pt}\right]\hspace{-2pt} Q_{u_{MC}}\Pi _{a}\bar{Q}_{r,1}^a\hspace{-2pt}\right)\hspace{-2pt} r_{y}  \notag \\
+\hspace{-2pt}\left[ \hspace{-2pt}
\begin{array}{c}
M \\ 
N%
\end{array}%
\hspace{-2pt}\right] \hspace{-2pt}Q_{u_{MC}}\Pi _{a}\bar{Q}_{r,2}^aQ_{u_{MC}}\left( \Pi _{a}\hspace{-2pt}-\hspace{-2pt}I\right)\hspace{-2pt} \Psi \hspace{-2pt}\left( \hspace{-2pt}\left[ 
\begin{array}{c}
u \\ 
y%
\end{array}%
\right] \hspace{-3pt}-\hspace{-3pt}\left[ \hspace{-2pt}
\begin{array}{c}
M \\ 
N%
\end{array}%
\hspace{-2pt}\right] \hspace{-2pt}v\hspace{-2pt}\right)\hspace{-2pt} 
\end{align*}
After some routine calculations,  we have 
\begin{gather}
\left[ 
\begin{array}{c}
u \\ 
y%
\end{array}%
\right] \hspace{-2pt}=\hspace{-2pt}\left[ 
\begin{array}{c}
M \\ 
N%
\end{array}%
\right]\hspace{-2pt}(I\hspace{-2pt}+\hspace{-2pt}Q_{u_{MC}}\Phi_1\bar{\Pi}_a) v\hspace{-2pt}+\hspace{-2pt}\left(\hspace{-2pt} \left[ \hspace{-2pt}
\begin{array}{c}
-\hat{Y} \\ 
\hat{X}%
\end{array}%
\hspace{-2pt}\right] \hspace{-2pt}+\hspace{-2pt}\left[ 
\begin{array}{c}
M \\ 
N%
\end{array}%
\right] \hspace{-2pt}\Phi_1\Phi_2 \right) \hspace{-2pt}r_{y},  \label{eq16-33} 
\end{gather}%
where 
\begin{align*}
&\Phi_1\hspace{-2pt}=\hspace{-2pt}Q_{u_{MC}}\hspace{-4pt}\left(\hspace{-2pt}I\hspace{-2pt}- \hspace{-2pt}\left(\Pi _{a}\bar{Q}_{r,2,a}^aQ_{u_{MC}}\left( \Pi _{a}\hspace{-2pt}-\hspace{-2pt}I\right)\hspace{-2pt}-\hspace{-2pt}I\right)\hspace{-2pt}\Psi \hspace{-3pt}\left[ \hspace{-2pt}
\begin{array}{c}
M \\ 
N%
\end{array}%
\hspace{-2pt}\right] \hspace{-3pt}Q_{u_{MC}}\hspace{-3pt}\right)^{\hspace{-2pt}-1}\\
&\Phi_2= \Pi_a\bar{Q}_{r,1}^a+\left(\Pi _{a}\bar{Q}_{r,2,a}^aQ_{u_{MC}}\left( \Pi _{a}-I\right)-I\right) \hspace{-2pt}\Psi \hspace{-2pt}\left[ \hspace{-2pt}
\begin{array}{c}
-\hat{Y} \\ 
\hat{X}%
\end{array}%
\hspace{-2pt}\right] \\
&\bar{\Pi}_a=\Pi _{a}\bar{Q}_{r,2}^a(I-Q_{u_{MC}}Q_{r,2})-Q_{r,2}-\Pi_aQ_{v}+Q_v\\
&\bar{Q}_{r,1}^a\hspace{-2pt}=\hspace{-2pt}(I\hspace{-2pt}-\hspace{-2pt}Q_{r,2}Q_{u_{MC}}\Pi _{a})^{-\hspace{-1pt}1}Q_{r,1},\bar{Q}_{r,2}^a\hspace{-2pt}=\hspace{-2pt}(I\hspace{-2pt}-\hspace{-2pt}Q_{u_{MC}}\Pi _{a})^{-\hspace{-1pt}1}\hspace{-1pt}Q_{r,2}
\end{align*}
It can be seen clearly that during attack-free
operations, i.e. $\Pi _{a}=I,$ %
  (\ref{eq16-33}) is exactly the nominal closed-loop dynamic, and the feedback
performance degradation caused by $\Pi _{a}\neq I$ is modelled by%
\begin{equation*}
\left[ 
\begin{array}{c}
u_{f} \\ 
y_{f}%
\end{array}%
\right] =\left( \left[ 
\begin{array}{c}
-\hat{Y} \\ 
\hat{X}%
\end{array}%
\right] +\left[ 
\begin{array}{c}
M \\ 
N%
\end{array}%
\right] \Phi_1\Phi_2 \right) r_{y}.
\end{equation*}%
Consequently, it holds 
\begin{equation*}
 \setlength{\abovedisplayskip}{5pt}\setlength{\belowdisplayskip}{5pt}
\begin{aligned}
r_{y,u}^{a} \hspace{-2pt}
&=\hspace{-2pt}(\bar{\Pi}_{a}\hspace{-2pt}+\hspace{-2pt}Q_{r,2}\hspace{-2pt}+\hspace{-2pt}\Phi_3)v
\hspace{-2pt}+\hspace{-2pt}\Pi _{a}\hspace{-2pt}\left( \hspace{-2pt}\bar{Q}_{r,1}^a\hspace{-2pt}+\hspace{-2pt}\Gamma  \Psi \hspace{-2pt}\left( \hspace{-2pt} \left[ \hspace{-2pt}
\begin{array}{c}
-\hspace{-2pt}\hat{Y} \\ 
\hat{X}%
\end{array}%
\hspace{-2pt}\right] \hspace{-4pt}+\hspace{-4pt}\left[ \hspace{-2pt}
\begin{array}{c}
M \\ 
N%
\end{array}%
\hspace{-2pt}\right] \hspace{-2pt}\Phi_1\hspace{-1pt}\Phi_2 \hspace{-2pt}\right) \hspace{-4pt}\right)\hspace{-3pt}r_y\\
\Gamma&=\bar{Q}_{r,2}^aQ_{u_{MC}}\left( \Pi
_{a}\hspace{-2pt}-\hspace{-2pt}I\right)\\
\Phi_3&=\Pi _{a}\bar{Q}_{r,2}^aQ_{u_{MC}}\hspace{-2pt}\left(\hspace{-2pt}(\Pi_a\hspace{-2pt}-\hspace{-2pt}I)Q_v\hspace{-2pt}-\hspace{-2pt}\left[ 
\begin{array}{c}
M \\ 
N%
\end{array}%
\right]\hspace{-3pt}Q_{u_{MC}}\Phi_1\bar{\Pi}_a\right)\hspace{-2pt}+\hspace{-2pt}Q_v
\end{aligned}%
\end{equation*}
Hence, residual signal%
\begin{equation}
 \setlength{\abovedisplayskip}{5pt}\setlength{\belowdisplayskip}{5pt}
r_{PDD}^a :=r_{y,u}^{a}-(Q_{r,2}+Q_v)v \label{eq-PFD2}
\end{equation}%
contains the full information about the feedback performance degradation and
can be used for the detection purpose. To this end, design $\Psi $ so that 
\begin{equation}
 \setlength{\abovedisplayskip}{5pt}\setlength{\belowdisplayskip}{5pt}
\gamma_{\Psi}:=\left\Vert \Psi \left( \left[ 
\begin{array}{c}
-\hat{Y} \\ 
\hat{X}%
\end{array}%
\right] +\left[ 
\begin{array}{c}
M \\ 
N%
\end{array}%
\right] \Phi_1\Phi_2\right) \right\Vert _{2}>>\left\Vert
\bar{Q}_{r,1}\right\Vert _{2},\label{eq-Psi}
\end{equation}%
where $\left\Vert \cdot \right\Vert _{2}$ denotes $\mathcal{H}_{2}$-norm of
a transfer function. On account of the fact that the cyber-attacks
target at reducing the system robustness leading to a larger $\gamma_{\Psi}$, 
it is expected that (\ref{eq-Psi}) holds.

Considering that $\Pi _{a}$ should be an inner, the condition of
attack stealthiness, the standard generalized likelihood ratio (GLR) detection schemes can be applied to detect the attacks. To be specific, during attack-free operation, the covariance matrix of $r_{PDD}$ can be
calculated based on the model%
\begin{equation*}
 \setlength{\abovedisplayskip}{5pt}\setlength{\belowdisplayskip}{5pt}
r_{PDD}=\left( \bar{Q}_{r,1}+\Psi \left( \left[ 
\begin{array}{c}
-\hat{Y} \\ 
\hat{X}%
\end{array}%
\right] +\left[ 
\begin{array}{c}
M \\ 
N%
\end{array}%
\right] Q_{u_{MC}}\bar{Q}_{r,1}\right) \right) r_{y}.
\end{equation*}%
To simplify the online implementation of the detection algorithm, a
post-filter $R(z)$ can be added so that 
\begin{equation*}
 \setlength{\abovedisplayskip}{5pt}\setlength{\belowdisplayskip}{5pt}
\bar{r}_{PDD}(z)=R(z)r_{PDD}(z)\sim \mathcal{N}\left( 0,\Sigma _{\bar{r}%
_{PDD}}\right) .
\end{equation*}%
Supposed that $N$ data, 
$\bar{r}_{PDD}^{a}\left( k+i\right) ,i=1,\cdots ,N,$ have been collected.
Using the GLR detection method, we have the
test statistic
\begin{equation*}
 \setlength{\abovedisplayskip}{5pt}\setlength{\belowdisplayskip}{5pt}
\begin{aligned}
&J =\frac{N}{2}\ln \frac{\det \left( \Sigma _{\bar{r}_{PDD}}\right) }{\det
\left( \hat{\Sigma}_{\bar{r}_{PDD}^{a}}\right) } \\
&\;\;\;\;+\hspace{-2pt}\frac{1}{2}\dsum\limits_{i=1}^{N}\hspace{-2pt}\left( \hspace{-2pt}\left( \bar{r}_{PDD}^{a}\left(
k\hspace{-2pt}+\hspace{-2pt}i\right) \right) ^{T}\hspace{-2pt}\left( \Sigma _{\bar{r}_{PDD}}^{-1}\hspace{-4pt}-\hspace{-2pt}\hat{\Sigma}_{%
\bar{r}_{PDD}^{a}}^{-1}\right) \bar{r}_{PDD}^{a}\left( k\hspace{-2pt}+\hspace{-2pt}i\right) \hspace{-2pt}\right)\hspace{-2pt} ,
\\
&\hat{\Sigma}_{\bar{r}_{PDD}^{a}} =\frac{1}{N}\dsum\limits_{i=1}^{N}\bar{r}%
_{PDD}^{a}\left( k+i\right) \left( \bar{r}_{PDD}^{a}\left( k+i\right)
\right) ^{T},
\end{aligned}
\end{equation*}
and determine the threshold accordingly. 

\subsection{Detection of stealthy attacks against tracking performance}

We next address the detection issue for the (additive) stealthy attacks subject to (%
\ref{eq-stealthy}). To this end, the residual (\ref{eq16-32}) is constructed and the controller is set to be 
\begin{equation*}
 \setlength{\abovedisplayskip}{5pt}\setlength{\belowdisplayskip}{5pt}
\begin{aligned}
&u=F\hat{x}+u_{MC}^{a}-Q_{u_{MC}}\bar{r}_{PD}, \\
&u_{MC} =Q_{u_{MC}}\left( r_{y,u}^{a}-Q_{r,2}v\right) +v.
\end{aligned}%
\end{equation*}
Next, we examine the system dynamics. It can be easily seen that for the stealthy attacks $Q_{r,2}a_{u_{MC}}\hspace{-2pt}+\hspace{-2pt}a_{r_{y,u}}\hspace{-4pt}=\hspace{-3pt}0$, we have
\begin{gather*}
 u=F\hat{x}+v+a_{u_{MC}}+\left( I-Q_{u_{MC}}Q_{r,2}\right)
^{-1}Q_{u_{MC}}Q_{r,1}r_{y},
\end{gather*}%
which results in 
\begin{align*}
&r_{y,u}^{a}
 =\Gamma r_{y}+Q_{r,2}v  +\hspace{-2pt}\Psi \hspace{-2pt}\left(\hspace{-2pt} \left[ \hspace{-2pt}
\begin{array}{c}
M \\ 
N%
\end{array}%
\hspace{-2pt}\right] \hspace{-2pt}\left( v\hspace{-2pt}+\hspace{-2pt}a_{u_{MC}}\right) \hspace{-2pt}\right)\\
&\Gamma  \hspace{-2pt}=\hspace{-2pt}\Psi \left( \hspace{-2pt}\left[ \hspace{-2pt}
\begin{array}{c}
-\hspace{-2pt}\hat{Y} \\ 
\hat{X}%
\end{array}%
\hspace{-2pt}\right] \hspace{-2pt}+\hspace{-2pt}\left[ \hspace{-2pt}
\begin{array}{c}
M \\ 
N%
\end{array}%
\hspace{-2pt}\right] Q_{u_{MC}}\bar{Q}_{r,1}\right) \hspace{-2pt}+\hspace{-2pt}\bar{Q}_{r,1}.  \notag
\end{align*}%
Thus, builds the following residual signal%
\begin{align}
r_{PDD}^a:&=r_{y,u}^{a}-\left( Q_{r,2}+Q_{v}\right) v
=\hspace{-2pt}Q_{v}a_{u_{MC}}\hspace{-2pt}+\hspace{-2pt}\Gamma r_{y}. \label{eq-PFD3} \end{align}%
Since $r_{y}\sim \mathcal{N}\left(
0,\Sigma _{r_{y}}\right) $, $\mathcal{X}_2$ test statistic can be applied for detecting $a_{u_{MC}}$ on the basis of model (\ref{eq-PFD3}).

\begin{remark}
In the paper, three different forms of the residual
signal  have been proposed, (\ref{eq15-38}), (\ref{eq-PFD2}) and (%
\ref{eq-PFD3}). While the first one is applied to the detection of
cyber-attacks in general, the latter two are dedicated to detecting stealthy
attacks. It is noteworthy that, independent of which of the three residual
forms is adopted, the overall closed-loop dynamics under attack-free
operations are identical with the nominal system dynamic. For a reliable
detection of cyber-attacks with high resilience, a schedule mechanism is to
be designed that manages switchings among the three detection algorithms.
For instance, run the detection algorithm based the residual $r_{y,u}$ (\ref%
{eq15-38}) as the base detection scheme, and switch to an algorithm with
residual  (\ref{eq-PFD2}) or (\ref{eq-PFD3}) over a short time
interval, and then return to the base detection scheme. 
 A meaningful side-effect of such a schedule mechanism is to
enhance the resilience against eavesdropping attacks, which is a typical
strategy followed by a cyber adversary to design stealthy attacks. 
\end{remark}

\section{An experimental study}
In this section,  a leader-follower robot system, as illustrated in Fig.~\ref{fig0}, is adopted to demonstrate our proposed methods. 
\subsection{System configuration}
\begin{figure}
	\begin{center}
		\includegraphics[width=0.31\textwidth]{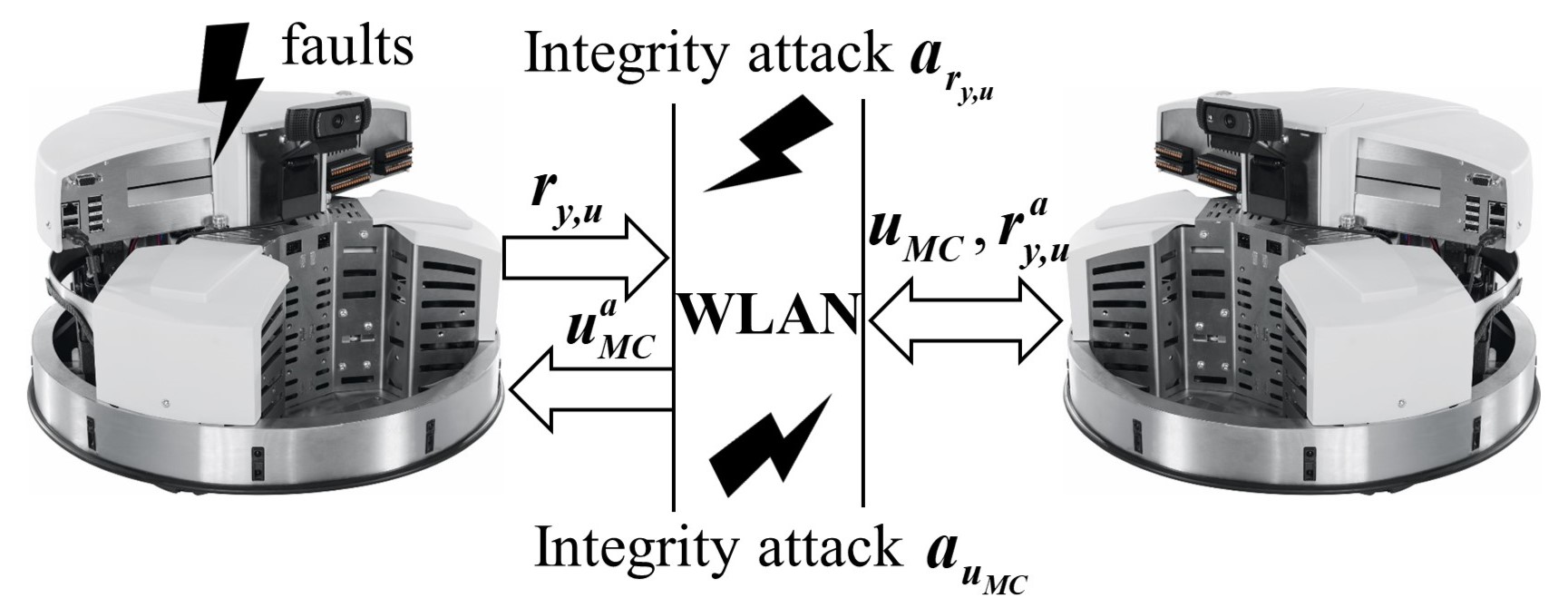} 
		\caption{The system configuration}  
		\label{fig0}                                 
	\end{center}                                 
\end{figure}
In this study, the leader generates the reference velocities, which are then tracked by the follower. In this context, the leader serves as the MC station for the follower.  Specifically, the follower sends the residual $r_{y, u}$ to the leader and receives the control input $u_{MC}$ from it, which contains the reference velocities. Their inputs and outputs are the setpoint speeds for three motors and three velocities, namely horizontal, vertical, and angular velocity. With the sampling rate of \SI{10}{Hz}, the state-space model for the follower is
\begin{align}
	&A_{f}\hspace{-4pt}=\hspace{-5pt}\begin{bmatrix}\hspace{-2pt}
		0.428& 0.020 & 0.0001\\
		0.026& 0.419 & 0.0037\\
		0.284 & -0.09 & 0.2922
	\hspace{-5pt}\end{bmatrix}\hspace{-4pt},\hspace{-2pt}
	B_{f}\hspace{-4pt}=\hspace{-5pt}\begin{bmatrix}\hspace{-2pt}
		-0.685 & 0.025 & 0.655\\
		0.406 & -0.803 & 0.344\\
		4.012 & 3.494 & 3.346
	\hspace{-2pt}\end{bmatrix}\hspace{-6pt}\times \hspace{-4pt}10^{\hspace{-1pt}-\hspace{-1pt}4}\notag
\end{align}
with $C_{f}=diag\{1,1,1 \},
	D_{f}=0$ and zero initial conditions. 
The noise covariance matrices are estimated based on the measurements. 
On the plant side, one observer, observer-based residual generator, and controller are constructed according to (\ref{eq15-38}). The state feedback gain $F$ is initialized as the LQ gain,
and the observer gain $L$ is set as the Kalman gain given in (\ref{eq-observer-ka}). 
The  parameters $Q_{r, 1}$, $Q_{r, 2}$, and $Q_{u_{MC}}$ are initialized as
\begin{align}
	&Q_{r, 1}=diag\left\{(z-0.1)^{-1}, (z-0.1)^{-1},(z-0.1)^{-1}\right\}\notag\\	
	&Q_{r, 2}=-0.15(z\hspace{-2pt}+\hspace{-2pt}0.1)^{-1}diag\left\{(z\hspace{-2pt}+\hspace{-2pt}0.4),(z\hspace{-2pt}+\hspace{-2pt}0.3),(z\hspace{-2pt}+\hspace{-2pt}0.2)\right\}\notag\\
	&Q_{u_{MC}}\hspace{-2pt}=\hspace{-2pt}diag\left\{\frac{10(z\hspace{-2pt}+\hspace{-2pt}0.1)}{z\hspace{-2pt}+\hspace{-2pt}0.4},  \frac{10(z\hspace{-2pt}+\hspace{-2pt}0.1)}{z\hspace{-2pt}+\hspace{-2pt}0.3},\frac{10(z\hspace{-2pt}+\hspace{-2pt}0.1)}{z\hspace{-2pt}+\hspace{-2pt}0.2}\right\}.\label{eq-qumc}
\end{align}
The achieved performance levels are
$
	\gamma_{\vartheta_{a}}=0.4000,
	\gamma_{r_y}=6.2406.
$
The feedforward controller $T$  is selected as 
\begin{equation}
 \setlength{\abovedisplayskip}{5pt}\setlength{\belowdisplayskip}{0pt}
	T=\begin{bmatrix}
		-4257.4943 & 2463.2315 & 662.2074\\
		10.9463 & -4940.2157 & 664.6608\\
		4284.4037 & 2462.1782 & 663.4313
	\end{bmatrix}.
\end{equation}

\subsection{Fault detection and fault-tolerant control}
\begin{figure}
	\begin{center}
		\includegraphics[width=0.45\textwidth]{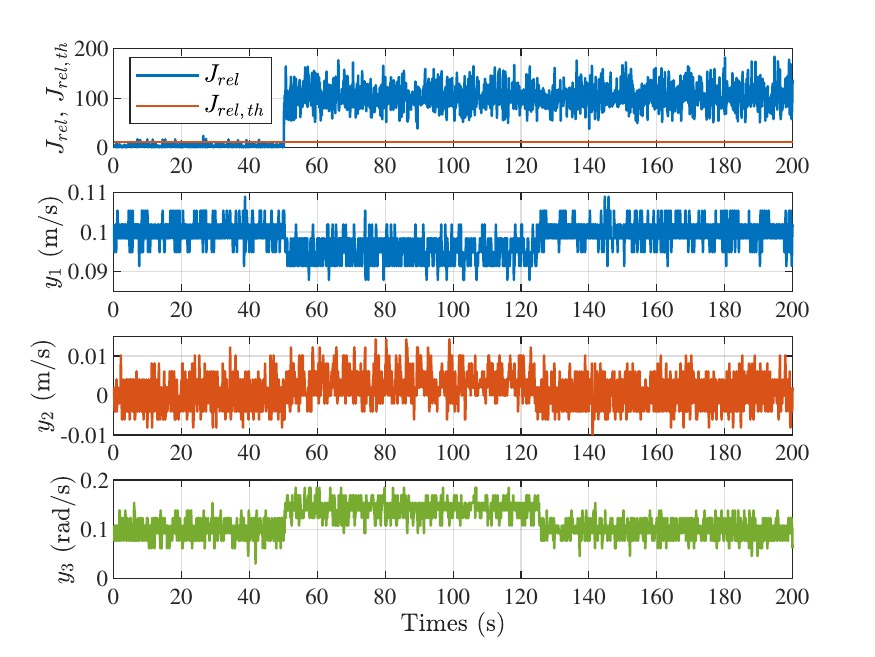} 
		\caption{Fault detection and fault-tolerant control}  
		\label{fig1}                                 
	\end{center}                                 
\end{figure}
For our purpose, the bias fault 
$f\sim\mathcal{N}\left(0.025, 1\times 10^{-6}\right)$ of the first sensor is  injected from the time instant $[50\text{s},125 \text{s}]$. 
Setting the FAR as $0.01$ leads to
$
	J_{rel, th} = 11.3450.
$
As shown in the first subfigure of Fig.~\ref{fig1}, the fault can be well detected.  
For FTC purpose, $Q_{r, 2}$ is reconfigured as
\begin{equation*}
 \setlength{\abovedisplayskip}{5pt}\setlength{\belowdisplayskip}{5pt}
Q_{r, 2}=-90(z\hspace{-2pt}+\hspace{-2pt}0.1)^{-1}diag\left\{(z\hspace{-2pt}+\hspace{-2pt}0.4),(z\hspace{-2pt}+\hspace{-2pt}0.3),(z\hspace{-2pt}+\hspace{-2pt}0.2)\right\}\notag
\end{equation*}
with $Q_{u_{MC}}$ given by (\ref{eq-qumc}). The associated performance level is
$
	\gamma_{\vartheta_{a}}=1.1099\times 10^{-3}.
$
Based on the obtained $Q_{r, 2}$ and $Q_{u_{MC}}$, $Q_{r, 1}$ is attained by solving (55) via the \emph{hinfsyn} command in \textsc{Matlab}. 
The achieved performance level is
$
	\gamma_{r_y}=1.0997.
$
For demonstration purpose, the controller is reconfigured in the time interval of $[\SI{125}{s}, \SI{200}{s}]$. According to the output velocity trajectories $y_1,y_2,y_3$ of Fig.~\ref{fig1}, the sensor fault can be well compensated by reconfiguring the fault-tolerant controller.

\subsection{Attack detection and attack-resilient control}
\begin{figure}
	\begin{center}
		\includegraphics[width=0.43\textwidth]{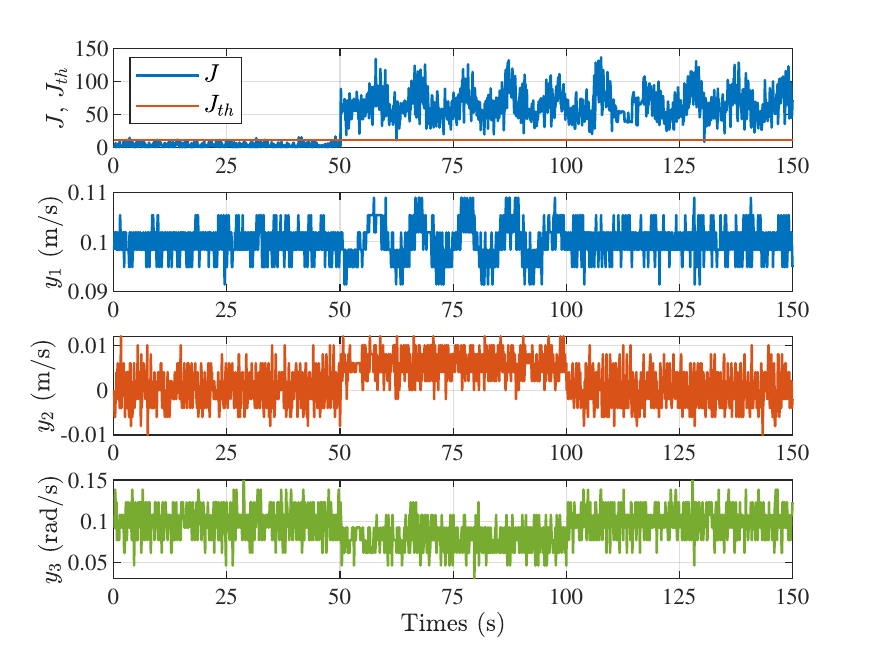} 
		\caption{Detection and resilient control of the cyber-attack}  
		\label{fig2}                                 
	\end{center}                                 
\end{figure}
In this study, the following attacks are injected into the corresponding channels during the time interval $[\SI{50}{s}, \SI{150}{s}]$
\begin{align}
	&a_{u_{MC}}(k)=\begin{bmatrix}
		a_{u_{MC, 1}}(k)&
		0&
		a_{u_{MC, 3}}(k)
	\end{bmatrix}^T\label{eq-attack1}\\
&		a_{u_{MC, 1}}(k)=0.05\sin(0.2\pi k), 
	a_{u_{MC, 3}}(k)=-0.05\sin(0.2\pi k)
	\notag\\
&	a_{r_{y, u}}(k)\hspace{-2pt}=\hspace{-2pt}\begin{bmatrix}
		0&
		a_{r_{y, u, 2}}(k)&
		0
	\end{bmatrix}^T,
	a_{r_{y, u, 2}}\hspace{-2pt}\sim\hspace{-2pt}\mathcal{N}\left(-0.025, 1\hspace{-2pt}\times \hspace{-2pt}10^{-10}\right)\hspace{-2pt}.\notag
\end{align}
Based on (\ref{eq16-14}), the post-filter $R$ for attack detection purpose is given.
Setting the FAR as $0.01$ leads to 
$
	J_{th} = 11.3450.
$
As shown in the first subfigure of Fig.~\ref{fig2}, these additive attacks can be well detected.  
Then, the controller is reconfigured during the time interval of $[\SI{100}{s}, \SI{150}{s}]$. As described in Fig.~\ref{fig2}, the attack-induced variations in robot velocities can be largely reduced, which verifies the effectiveness of the resilient controller.

\subsection{Detection of attack switching-on and switching-off}

\begin{figure}
	\begin{center}
		\includegraphics[width=0.43\textwidth]{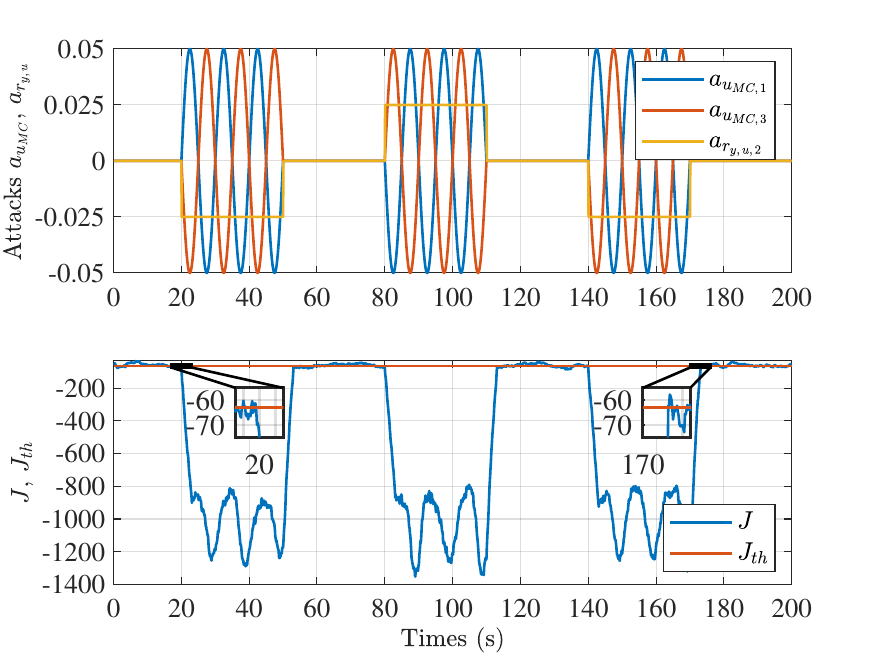} 
		\caption{Detection of attack switching-on and switching-off}  
		\label{fig3}                                 
	\end{center}                                 
\end{figure}
For our purpose, the parameters $s$, $\gamma$, and $L_0$ are selected as
$
	s = 30,
	\gamma = 500,
	L_0 = 1\times 10^{-4},
	\tau = 530.
$
Then, the lower and upper bound $L_l$ and $L_u$  are calculated as
$
	L_l = 0.1433,
	L_u = 1.0474.
$
According to the non-central $\chi^{2}$ distribution table, setting
$
	\left\lVert \bar{r}_{\eta_a, s}(k_0) \right\rVert^{2}=129.15
$
gives the FAR as $0.01$. Consequently, the thresholds for the evaluation functions in (50) are given by
\begin{equation}
 \setlength{\abovedisplayskip}{5pt}\setlength{\belowdisplayskip}{5pt}
	J_{th}=\begin{cases}
		53.2207, & \left\lVert \bar{r}_{\eta_a, s}(k_0) \right\rVert^{2}\leq L_{l}^{2},\\
		-9.7366, & L_{l}^{2}\leq\left\lVert \bar{r}_{\eta_a, s}(k_0) \right\rVert^{2}\leq L_{u}^{2},\\
		-62.9573, & \left\lVert \bar{r}_{\eta_a, s}(k_0) \right\rVert^{2}\geq L_{u}^{2}.
	\end{cases}
\end{equation}
The periodic attacks $a_{u_{MC}}$ and $a_{r_{y, u}}$ are assumed to be launched in $[\SI{20}{s}, \SI{170}{s}]$, as shown in  Fig.~\ref{fig3}. According to the experimental results in the second subfigure of Fig.~\ref{fig3}, the detector detects the attack switching-on right after $\SI{20}{s}$, while the attack switching-off is detected with a delay around $\SI{3}{s}$. 

\subsection{Comparison study}
\begin{figure}
	\begin{center}
		\includegraphics[width=0.43\textwidth]{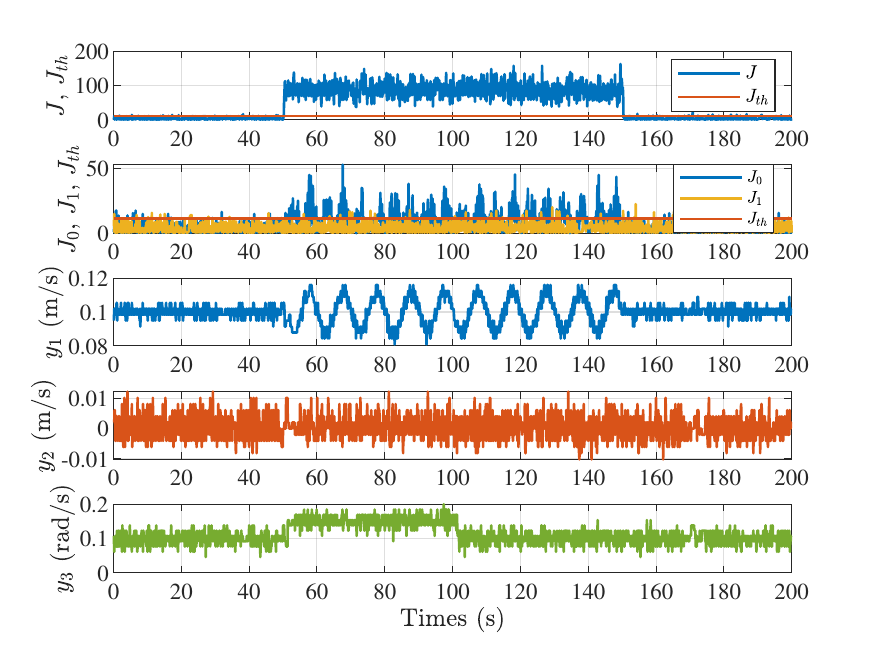} 
		\caption{Detection and resilient control of the traditional CPS configuration}  
		\label{fig4}                                 
	\end{center}                                 
\end{figure}
In this section, we compare our proposed scheme with the traditional CPS configuration  (\ref{eq15-1a})-(\ref{eq14-0a}). Here, the Kalman filter-based $\chi^{2}$ detector is utilized for attack detection. The controller in the MC station is implemented as the Youla parameterization-based controller in the observer form. 
The default $Q$ in (10) is set as $Q=Q_{r_1}$ in (\ref{eq-qumc}). 
For a fair comparison, attacks $a_{u_{MC}}$ and $a_y$ are set the same as those in (34) and (37).
Firstly, we would like to compare the attack detection performance. As demonstrated in the first subfigure of Fig.~\ref{fig4}, these attacks can be well detected by the $\chi^{2}$ detector on the MC side. 
We next consider the cyber-attack (\ref{eq-attack1}) with $a_{r_{y,u}}=0$.  
The comparison results are shown in the second subfigure of Fig.~\ref{fig4} in which $J_0$ and $J_1$ are evaluation functions for our attack detector and the $\chi^{2}$ detector, respectively. It can be seen that slight input attacks can be detected by our method, while they remain undetectable to the traditional $\chi^{2}$ detector. 

Then, we would like to compare this traditional configuration with our proposed scheme regarding attack resilience performance. In the time interval of $[\SI{100}{s}, \SI{150}{s}]$, the parameter $Q$ is reconfigured by solving the model matching problem as
\begin{equation}
 \setlength{\abovedisplayskip}{2pt}\setlength{\belowdisplayskip}{2pt}
\min_{Q}\left\Vert \left[ 
\begin{array}{c}
-\hat{Y} \\ 
\hat{X}%
\end{array}%
\right] +\left[ 
\begin{array}{c}
M \\ 
N%
\end{array}%
\right] Q\right\Vert _{\infty }.
\end{equation}%
The corresponding resilient control performance is illustrated in Fig.~\ref{fig4}, which shows the advantage of our resilient control. 

\subsection{Detection of additive and multiplicative stealthy attacks}
\begin{figure}
	\begin{center}
		\includegraphics[width=0.47\textwidth]{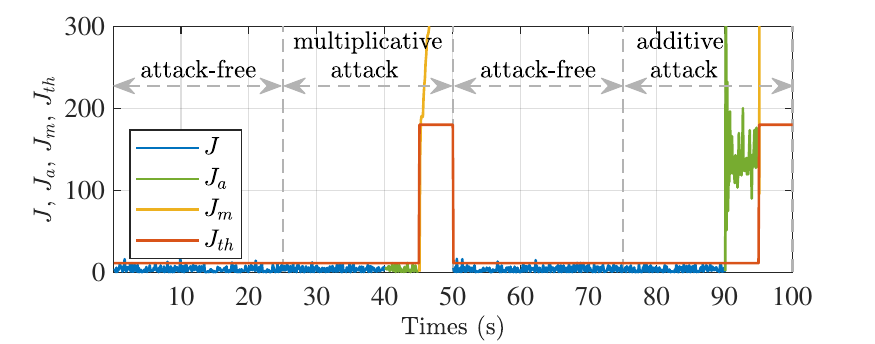} 
		\caption{Detection of additive and multiplicative stealthy attacks}  
		\label{fig5}                                 
	\end{center}                                 
\end{figure}
The additive stealthy attacks are designed according to (57). Specifically, attack $a_{u_{MC}}$ is set as
\begin{equation*}
 \setlength{\abovedisplayskip}{5pt}\setlength{\belowdisplayskip}{5pt}
 \begin{aligned}
&	a_{u_{MC}}(k)\hspace{-2pt}=\hspace{-2pt}\begin{bmatrix}
		a_{u_{MC,}}^1(k)&
		0&
		a_{u_{MC}}^3(k)
	\end{bmatrix}^T\hspace{-2pt}
\\
	&a_{u_{MC}}^1\hspace{-2pt}(k)\hspace{-2pt}=\hspace{-2pt}0.05\sin(0.2\pi k)\hspace{-2pt}+\hspace{-2pt}1,a_{u_{MC}}^3\hspace{-2pt}(k)\hspace{-2pt}=\hspace{-2pt}-\hspace{-2pt}0.05\sin(0.2\pi k)\hspace{-2pt}+\hspace{-2pt}1.
\end{aligned}
\end{equation*}
The attacks $a_{r_{y, u}}$ is generated accordingly. To detect the additive stealthy attacks, the residual $r_{y, u}^{a}$ is constructed as (\ref{eq-PFD3}) with the parameter $\Psi$ is set as
$
	\Psi = \begin{bmatrix}
		5\times 10^{-2}I & I
	\end{bmatrix}. 
$
Similar as before, a post-filter can be designed for detection purpose and the threshold is set as $J_{th} = 11.3450$. 

On the other hand, the multiplicative stealthy attacks are designed according to (\ref{eq-sa-de}). In this study, we set
$
	\Pi_{a} = -I.
$
For detection purpose, the residual $r_{y,u}$ is built according to (\ref{eq16-32}) with the same $\Psi$.
We use a numerical approach to determining the threshold, which returns
$
	J_{th} = 180.
$
In this study, the detector switching period is set as $\SI{25}{s}$, where regular, additive, and multiplicative attack detector take $\SI{15}{s}$, $\SI{5}{s}$, and $\SI{5}{s}$, respectively with
$J$, $J_a$, and $J_m$ as the associated test statistic.   Besides, multiplicative stealthy attacks are lunched in $[\SI{25}{s}, \SI{50}{s}]$, while additive ones are injected in $[\SI{75}{s}, \SI{100}{s}]$. According to Fig. \ref{fig5}, regular attack detector can detect neither the additive nor multiplicative stealthy attacks. The additive and multiplicative stealthy attacks can be detected by their dedicated detectors. 


\section{Conclusions}
This paper is mainly concentrated on the integrated design of detection and  resilient control schemes for both the process faults and cyber-attacks for CPSs.
To this end, in the first part of this paper, the dynamics of a general type of CPSs under defined types of cyber-attacks have been studied, and the possible impairment of system dynamics have been analyzed. It can be concluded that the cyber-attacks solely affect the image subspace of the process, while the process faults lead to the change exclusively in the image subspace of the controller. 
 It has been
delineated that the capability of the standard observer-based detection and FTC schemes are strongly limited in dealing with the detection and resilient control for cyber-attacks. 
To handle this issue, an alternative representation of the closed-loop dynamics has been studied, 
which reveals that i) I/O signal space consists of two complementary subspaces as the image subspaces of the plant and the controller, ii) there exists a one-to-one mapping between I/O signals and the I/O residuals. These observations motivate the second part of this paper to investigate the modified CPS configuration by transmitting $r_{y,u}$ as a fusion of the I/O residuals for both the detection and resilient control purpose. 

The further efforts have been dedicated to develop the associated detection schemes for fault detection, cyber-attack detection and resilient FTC.  It has been shown that not only the high resilience  under cyber-attacks and high fault-tolerance against process faults is ensured, but also the ``fail-safe" cyber-security and data privacy are guaranted with limited online computation and communication effort.
To further enhance the security of the proposed CPS configuration, the design schemes of the stealthy attacks against both the tracking behaviour and feedback control performance have been developed from the attacker's point of view. It is followed by the performance-based detection schemes from the defender's point of view. 



\begin{thebibliography}{10}
\providecommand{\url}[1]{#1}
\csname url@samestyle\endcsname
\providecommand{\newblock}{\relax}
\providecommand{\bibinfo}[2]{#2}
\providecommand{\BIBentrySTDinterwordspacing}{\spaceskip=0pt\relax}
\providecommand{\BIBentryALTinterwordstretchfactor}{4}
\providecommand{\BIBentryALTinterwordspacing}{\spaceskip=\fontdimen2\font plus
\BIBentryALTinterwordstretchfactor\fontdimen3\font minus
  \fontdimen4\font\relax}
\providecommand{\BIBforeignlanguage}[2]{{%
\expandafter\ifx\csname l@#1\endcsname\relax
\typeout{** WARNING: IEEEtran.bst: No hyphenation pattern has been}%
\typeout{** loaded for the language `#1'. Using the pattern for}%
\typeout{** the default language instead.}%
\else
\language=\csname l@#1\endcsname
\fi
#2}}
\providecommand{\BIBdecl}{\relax}
\BIBdecl

\bibitem{Ding2020}
S.~X. Ding, \emph{Advanced Methods for Fault Diagnosis and Fault-tolerant
  Control}.\hskip 1em plus 0.5em minus 0.4em\relax Berlin: Springer-Verlag,
  2020.

\bibitem{Ding2014}
------, \emph{Data-Driven Design of Fault Diagnosis and Fault-Tolerant Control
  Systems}.\hskip 1em plus 0.5em minus 0.4em\relax London: Springer-Verlag,
  2014.

\bibitem{CRB2001}
L.~H. Chiang, E.~L. Russell, and R.~D. Braatz, \emph{Fault Detection and
  Diagnosis in Industrial Systems}.\hskip 1em plus 0.5em minus 0.4em\relax
  London: Springer, 2001.

\bibitem{PFC00}
R.~J. Patton, P.~M. Frank, and R.~N.~C. (Eds.), \emph{Issues of Fault Diagnosis
  for Dynamic Systems}.\hskip 1em plus 0.5em minus 0.4em\relax London:
  Springer, 2000.

\bibitem{AET_SMO_book_2011}
H.~Alwi, C.~Edwards, and C.~P. Tan, \emph{Fault Detection and Fault-Tolerant
  Control Using Sliding Modes}.\hskip 1em plus 0.5em minus 0.4em\relax
  Springer-Verlag, 2011.

\bibitem{DeruiDing2021-survey}
D.~Ding, Q.-L. Han, X.~Ge, and J.~Wang, ``Secure state estimation and control
  of cyber-physical systems: A survey,'' \emph{IEEE Trans. Syst., Man, and
  Cybern.: Syst.}, vol.~51, no.~1, pp. 176--190, 2021.

\bibitem{TGXHV2020}
S.~{Tan}, J.~M. {Guerrero}, P.~{Xie}, R.~{Han}, and J.~C. {Vasquez}, ``Brief
  survey on attack detection methods for cyber-physical systems,'' \emph{IEEE
  Systems Journal}, vol.~14, pp. 5329--5339, 2020.

\bibitem{TEIXEIRA-zero-attack_2015}
A.~Teixeira, I.~Shames, H.~Sandberg, and K.~H. Johansson, ``A secure control
  framework for resource-limited adversaries,'' \emph{Automatica}, vol.~51, pp.
  135 -- 148, 2015.

\bibitem{Zhou2021IEEE-Proc}
C.~Zhou, B.~Hu, Y.~Shi, Y.-C. Tian, X.~Li, and Y.~Zhao, ``A unified
  architectural approach for cyberattack-resilient industrial control
  systems,'' \emph{Proceedings of the IEEE}, vol. 109, pp. 517--541, 2021.

\bibitem{Survey-attack-detection2018}
J.~Giraldo, D.~Urbina, A.~Cardenas, J.~Valente, M.~Faisal, J.~Ruths, N.~O.
  Tippenhauer, H.~Sandberg, and R.~Candell, ``A survey of physics-based attack
  detection in cyber-physical systems,'' \emph{ACM Comput. Surv.}, vol.~51,
  2018.

\bibitem{Zhang2020attacks}
Q.~Zhang, K.~Liu, Y.~Xia, and A.~Ma, ``Optimal stealthy deception attack
  against cyber-physical systems,'' \emph{IEEE Trans. Cybern.}, vol.~50, no.~9,
  pp. 3963--3972, 2020.

\bibitem{DIBAJI2019-survey}
S.~M. Dibaji, M.~Pirani, D.~B. Flamholz, A.~M. Annaswamy, K.~H. Johansson, and
  A.~Chakrabortty, ``A systems and control perspective of {CPS} security,''
  \emph{Annual Reviews in Control}, vol.~47, p. 394–411, 2019.

\bibitem{Chen2018attacks}
Y.~Chen, S.~Kar, and J.~M.~F. Moura, ``Cyber-physical attacks with control
  objectives,'' \emph{IEEE Trans. Autom. Control}, vol.~63, no.~5, pp.
  1418--1425, 2018.

\bibitem{Yang2022attacks}
X.-X. Ren and G.~H. Yang, ``Kullback–leibler divergence-based optimal
  stealthy sensor attack against networked linear quadratic gaussian systems,''
  \emph{IEEE Trans. Cybern.}, vol.~52, no.~11, pp. 11\,539--11\,548, 2022.

\bibitem{Shang2022attacks}
J.~Shang, D.~Cheng, J.~Zhou, and T.~Chen, ``Asymmetric vulnerability of
  measurement and control channels in closed-loop systems,'' \emph{IEEE Trans.
  Contr. Netw. Syst.}, vol.~9, no.~4, pp. 1804--1815, 2022.

\bibitem{MMM2020}
A.~M. Mohan, N.~Meskin, and H.~Mehrjerdi, ``A comprehensive review of the
  cyber-attacks and cyber-security on load frequency control of power
  systems,'' \emph{Energies}, vol.~13, 2020.

\bibitem{DLautomatica2022}
S.~X. Ding, L.~Li, D.~Zhao, C.~Louen, and T.~Liu, ``Application of the unified
  control and detection framework to detecting stealthy integrity cyber-attacks
  on feedback control systems,'' \emph{Automatica}, vol. 142, p. 110352, 2022.

\bibitem{HDCL-TCSII-2024}
Y.~Hu, X.~Dai, D.~Cui, and Q.~Liu, ``Anomaly identification for cyber-physical
  systems subject to replay attacks and sensor faults,'' \emph{IEEE Trans.
  Circuits and Systems II: Express Briefs}, 2024, doi:
  10.1109/TCSII.2024.3349777.

\bibitem{ZKPP-CSL-2021}
K.~Zhang, C.~Keliris, T.~Parisini, and M.~M. Polycarpou, ``Identification of
  sensor replay attacks and physical faults for cyber-physical systems,''
  \emph{IEEE Control Systems Letters}, vol.~6, pp. 1178--1183, 2021.

\bibitem{RA-EJC-2023}
M.~Ramadan and F.~Abdollahi, ``An active approach for isolating replay attack
  from sensor faults,'' \emph{European Journal of Control}, vol.~69, p. 100725,
  2023.

\bibitem{SZCDCX-AUTO-2024}
S.~Shen, C.~Zhang, R.~Chai, L.~Dai, S.~Chai, and Y.~Xia, ``Stabilizing
  nonlinear model predictive control under {Denial-of-Service} attack via
  dynamic samples selection,'' \emph{Automatica}, vol. 164, p. 111591, 2024.

\bibitem{AY-IS-2018}
L.~An and G.~Yang, ``Improved adaptive resilient control against sensor and
  actuator attacks,'' \emph{Information Sciences}, vol. 423, pp. 145--156,
  2018.

\bibitem{YZDYF-AUTO-2023}
Z.~Ye, D.~Zhang, C.~Deng, H.~Yan, and G.~Feng, ``Finite-time resilient sliding
  mode control of nonlinear {UMV} systems subject to {DoS} attacks,''
  \emph{Automatica}, vol. 156, p. 111170, 2023.

\bibitem{Ferreira-survey-2024}
\BIBentryALTinterwordspacing
M.~Segovia-Ferreira, J.~Rubio-Hernan, A.~Cavalli, and J.~Garcia-Alfaro, ``A
  survey on cyber-resilience approaches for cyber-physical systems,'' vol.~56,
  no.~8, 2024. [Online]. Available: \url{https://doi.org/10.1145/3652953}
\BIBentrySTDinterwordspacing

\bibitem{SMMSF-TAC-2021}
T.~Sui, Y.~Mo, D.~Marelli, X.~Sun, and M.~Fu, ``The vulnerability of
  cyber-physical system understealthy attacks,'' \emph{IEEE Trans. Autom.
  Control}, vol.~66, no.~2, pp. 637--650, 2021.

\bibitem{ZHOU2023110723}
J.~Zhou, J.~Shang, and T.~Chen, ``Optimal deception attacks on remote state
  estimators equipped with interval anomaly detectors,'' \emph{Automatica},
  vol. 148, p. 110723, 2023.

\bibitem{QLSY-TAC-2018}
J.~Qin, M.~Li, L.~Shi, and X.~Yu, ``Optimal denial-of-service attack scheduling
  with energy constraint over packet-dropping networks,'' \emph{IEEE Trans.
  Autom. Control}, vol.~63, no.~6, pp. 1648--1663, 2018.

\bibitem{AY-TAC-2018}
L.~An and G.~H. Yang, ``Data-driven coordinated attack policy design based on
  adaptive {$\mathcal{L}_2$}-gain optimal theory,'' \emph{IEEE Trans. Autom.
  Control}, vol.~63, no.~6, pp. 1850--1857, 2018.

\bibitem{GSJS-AUTO-2018}
Z.~Guo, D.~Shi, K.~H. Johansson, and L.~Shi, ``Worst-case stealthy
  innovation-based linear attack on remote state estimation,''
  \emph{Automatica}, vol.~89, pp. 117--124, 2018.

\bibitem{RWDS-AUTO-22018}
X.~Ren, J.~Wu, S.~Dey, and L.~Shi, ``Attack allocation on remote state
  estimation in multi-systems: Structure results and asymptotic solution,''
  \emph{Automatica}, vol.~87, pp. 184--194, 2018.

\bibitem{WSC-TC-2018}
G.~Wu, J.~Sun, and J.~Chen, ``Optimal data injection attacks in cyber-physical
  systems,'' \emph{IEEE Trans. Cybern.}, vol.~48, no.~12, pp. 3302--3312, 2018.

\bibitem{FQCTZ-TAC-2020}
C.~Fang, Y.~Qi, J.~Chen, R.~Tan, and W.~Zheng, ``Stealthy actuator signal
  attacks in stochastic control systems: {P}erformance and limitations,''
  \emph{IEEE Trans. Autom. Control}, vol.~65, no.~9, pp. 3927--3934, 2020.

\bibitem{GYH-TCNS-2021}
L.~Guo, H.~Yu, and F.~Hao, ``Optimal allocation of false data injection attacks
  for networked control systems with two communication chanels,'' \emph{IEEE
  Trans. Control Netw. Syst.}, vol.~8, no.~1, pp. 2--14, 2021.

\bibitem{ZHZL-TC-2021}
Z.~Zhao, Y.~Huang, Z.~Zhen, and Y.~Li, ``Data-driven false data-injection
  attack design and detection in cyber-physical systems,'' \emph{IEEE Trans.
  Cybern.}, vol.~51, no.~12, pp. 6179--6187, 2021.

\bibitem{JY-TC-2024}
K.~Jin and D.~Ye, ``Optimal innovation-basedstealthy attacks in networked {LQG}
  systems with attackcost,'' \emph{IEEE Trans. Cybern.}, vol.~54, pp. 787--796,
  2024.

\bibitem{SCZC-TCNS-2022}
J.~Shang, D.~Cheng, J.~Zhou, and T.~Chen, ``Asymmetric vulnerability of
  measurement and control channels in closed-loop systems,'' \emph{IEEE Trans.
  Contr. Net. Syst.}, vol.~9, pp. 1804--1815, 2022.

\bibitem{ZLTLCX-TAC-2023}
Q.~Zhang, K.~Liu, A.~M.~H. Teixeira, Y.~Li, S.~Chai, and Y.~Xia, ``On online
  kullback-leibler divergence-based stealthy attack against cyber-physical
  systems,'' \emph{IEEE Trans. Autom. Control}, vol.~68, pp. 3672--3679, 2023.

\bibitem{LY-TAC-2022}
A.~Lu and G.-H. Yang, ``False data injection attacks against state estimation
  without knowledge of estimators,'' \emph{IEEE Trans. Autom. Control},
  vol.~67, no.~9, pp. 4529--4540, 2022.

\bibitem{MS-TAC-2016}
Y.~Mo and B.~Sinopoli, ``On the performance degradation of cyber-physical
  systems under stealthy integrity attacks,'' \emph{IEEE Trans. Autom.
  Control}, vol.~61, pp. 2618--2624, 2016.

\bibitem{Zhou96}
K.~Zhou, J.~Doyle, and K.~Glover, \emph{Robust and Optimal Control}.\hskip 1em
  plus 0.5em minus 0.4em\relax Upper Saddle River, New Jersey: Prentice-Hall,
  1996.

\bibitem{MT-method-IEEE-TAC2020}
P.~{Griffioen}, S.~{Weerakkody}, and B.~{Sinopoli}, ``A moving target defense
  for securing cyber-physical systems,'' \emph{IEEE Trans. Autom. Control},
  vol.~66, pp. 2016--2031, 2021.

\bibitem{LWZCY-AUTO-2023}
T.~Li, Z.~Wang, L.~Zou, B.~Chen, and L.~Yu, ``A dynamic encryption–decryption
  scheme for replay attack detectionin cyber–physical systems,''
  \emph{Automatica}, vol. 151, p. 110926, 2023.

\bibitem{Zhou_review_2020}
D.~Zhou, Y.~Zhao, Z.~Wang, X.~He, and M.~Gao, ``Review on diagnosis techniques
  for intermittent faults in dynamic systems,'' \emph{IEEE Trans. on Indus.
  Electronics}, vol.~67, pp. 2337 -- 2347, 2020.

\bibitem{Smith2015}
R.~S. {Smith}, ``Covert misappropriation of networked control systems:
  Presenting a feedback structure,'' \emph{IEEE Control Systems Magazine},
  vol.~35, pp. 82--92, 2015.

\bibitem{ZhangTAC2023}
Q.~Zhang, K.~Liu, A.~M.~H. Teixeira, Y.~Li, S.~Chai, and Y.~Xia, ``An online
  kullback–leibler divergence-based stealthy attack against cyber-physical
  systems,'' \emph{IEEE Trans. Autom. Control}, vol.~68, no.~6, pp. 3672--3679,
  2023.

\bibitem{LIU2022110079}
H.~Liu, Y.~Ni, L.~Xie, and K.~H. Johansson, ``How vulnerable is
  innovation-based remote state estimation: Fundamental limits under linear
  attacks,'' \emph{Automatica}, vol. 136, p. 110079, 2022.

\bibitem{REN2023110895}
X.-X. Ren, G.-H. Yang, and X.-G. Zhang, ``Optimal stealthy attack with
  historical data on cyber–physical systems,'' \emph{Automatica}, vol. 151,
  p. 110895, 2023.

\bibitem{Jin2024}
K.~Jin and D.~Ye, ``Optimal innovation-based stealthy attacks in networked lqg
  systems with attack cost,'' \emph{IEEE Trans. Cybern.}, vol.~54, no.~2, pp.
  787--796, 2024.

\end{thebibliography}
\end{document}